\documentclass[runningheads]{llncs}

\newif\ifappendix
\appendixtrue

\ifappendix
\else
\def\marginpar#1{}
\fi

\usepackage{amsmath,amssymb}
\usepackage{mathtools}  
\usepackage[vlined]{algorithm2e}
\usepackage[hypertexnames=false]{hyperref}
\usepackage{standalone}
\usepackage{enumitem}
\usepackage{booktabs}
\usepackage{longtable}
\usepackage{xcolor}
\usepackage{colortbl}
\usepackage{todonotes}
\usepackage{fontawesome5}
\usepackage{adjustbox}
\usepackage{xspace}
\usepackage{spotltl}
\usepackage{mathrsfs}           
\usepackage{thmtools}
\usepackage{thm-restate}
\usepackage{subcaption}
\usepackage{tikz}
\usetikzlibrary{automata}
\usetikzlibrary{arrows.meta}
\usetikzlibrary{bending}
\usetikzlibrary{calc}
\usetikzlibrary{positioning}
\usetikzlibrary{shapes.geometric}
\usetikzlibrary{decorations.text}
\usetikzlibrary{backgrounds}
\usetikzlibrary{myautomata}

\tikzset{
      >={Stealth[round,bend]},
      StealthDouble/.tip={Stealth[round,quick,length=5.75pt]},
      double>/.style={-StealthDouble,double},
      root/.style={draw,rectangle,rounded corners=3pt},
      inode/.style={draw,circle,inner sep=0pt,minimum height=1.4em},
      low/.style={->,densely dotted},
      high/.style={->},
      glink/.style={->},
      killed/.style={red!50},
      scca/.style={draw=cyan,rounded corners=1mm},
      sccr/.style={draw=magenta,rounded corners=1mm},
      termn/.style={draw,rectangle},
      terma/.style={draw,rectangle,double},
      rlink/.style={->,shorten >=3pt},
      win/.style={fill=green!20},
      winedge/.style={thick,draw=green!66!black},
      lose/.style={fill=red!20},
      control/.style={draw,inner sep=1pt,diamond,aspect=1},
      env/.style={draw,inner sep=1pt,rectangle,minimum width=1.2em},
      automaton/.style={
        shorten >=1pt,>={Stealth[round,bend]},
        node distance=2cm,
        initial text=,
        every initial by arrow/.style={every node/.style={inner sep=0pt}},
      every state/.style={
        align=center,
        fill=white,
        minimum size=7.5mm,
        inner sep=0pt,
        execute at begin node=\strut,
      },%
   }
}

\newcommand{\tr}{\mathsf{tr}}
\def\lang{\mathscr{L}}

\SetProcNameSty{texttt}
\SetKwProg{Fn}{Function}{}{end}
\SetVlineSkip{0.0em}
\SetKwComment{tcp}{// }{}
\SetKwInOut{Input}{input}
\SetKwInOut{Output}{output}
\SetKwInput{Var}{var}
\SetKwFunction{mk}{makebdd}
\SetKwFunction{apply}{apply2}
\SetKwFunction{applyone}{apply1}
\SetKwFunction{applysc}{apply2sc}
\SetKwFunction{eval}{eval}
\SetKwFunction{fun}{fun}
\SetKwFunction{leaffun}{leaffun}
\SetKwFunction{funsh}{f}
\SetKwFunction{qtb}{quant\_to\_bool}
\SetKwFunction{leaves}{leaves}
\SetKwFunction{solveone}{solve1}
\SetKwFunction{solvetwo}{solve2}
\SetKwFunction{newvertex}{new\_vertex}
\SetKwFunction{newedge}{new\_edge}
\SetKwFunction{freezevertex}{freeze\_vertex}
\SetKwFunction{setwinner}{set\_winner}
\SetKwFunction{realizability}{realizability}
\SetKwFunction{declarevertex}{declare\_vertex}
\SetKwFunction{refine}{refine}
\SetKwFunction{win}{win\textsubscript{\textsc{o}}}
\SetKwFunction{lose}{win\textsubscript{\textsc{i}}}
\SetKwFunction{dfs}{dfs}
\SetKw{Assert}{assert}
\SetKw{SuchThat}{such that}
\SetKw{Continue}{continue}
\SetKw{Break}{break}

\newcommand{\CA}{\mathcal{A}}

\newcommand{\CI}{\mathcal{I}}
\newcommand{\CO}{\mathcal{O}}
\newcommand{\CS}{\mathcal{S}}

\newcommand{\CP}{\mathcal{P}}
\newcommand{\CQ}{\mathcal{Q}}
\newcommand{\BB}{\mathbb{B}}
\newcommand{\NN}{\mathbb{N}}

\newcommand{\MTBDD}{\mathsf{MTBDD}}

\newcommand{\LTL}{{\texorpdfstring{\ensuremath{\mathsf{LTL}}}{LTL}}\xspace}
\newcommand{\LTLf}{{\texorpdfstring{LTL\textsubscript{f}}{LTLf}}\xspace}

\definecolor{lime}{HTML}{A6CE39}
\DeclareRobustCommand{\orcidicon}{
	\hspace{-2.5mm}
	\begin{tikzpicture}[baseline={(0,-0.12)}]
	\draw[lime, fill=lime] (0,0)
	circle [radius=0.16]
	node[white] (ID) {{\fontfamily{qag}\selectfont \tiny ID}};
	\draw[white, fill=white] (-0.0625,0.095)
	circle [radius=0.007];
	\end{tikzpicture}
	\hspace{-2.5mm}
      }
\def\orcidID#1{\href{https://orcid.org/#1}{\smash{\orcidicon}}}

\title{Fast Obligation Translation and Synthesis}

\author{Alexandre~Duret-Lutz\inst{1}\orcidID{0000-0002-6623-2512} \and
  Giuseppe~De~Giacomo\inst{2}\orcidID{0000-0001-9680-7658} \and
  Marcin Jurdzinski\inst{3}\orcidID{0000-0003-3640-8481} \and
  Nir~Piterman\inst{4}\orcidID{0000-0002-8242-5357} \and
  Moshe~Y.~Vardi\inst{5}\orcidID{0000-0002-0661-5773} \and
  Shufang~Zhu\inst{6}\orcidID{0000-0002-5922-8750}
}

\authorrunning{A. Duret-Lutz et al.}
\institute{LRE, EPITA, Le Kremlin-Bicêtre, France \and
  University of Oxford, Oxford, UK \and
  DIMAP, University of Warwick,  Coventry, UK \and
  \!\!\mbox{University\,of\,Gothenburg\,and\,Chalmers\,University\,of\,Technology,\,Gothenburg,\,Sweden} \and
  Rice University, Houston, Texas, USA \and
  University of Liverpool, Liverpool, UK
}

\begin{document}

\maketitle

\begin{abstract}
  Syntactic obligations are a fragment of LTL formulas that translate to
  deterministic weak $\omega$-automata (DWA).
  We show that syntactic obligations can be very efficiently converted to minimal
  DWA represented using multi-terminal binary decision diagrams
  (MTBDDs), and that synthesis of such specifications can be solved
  directly on the MTBDD representation on the fly.
  Our implementation in Spot shows substantial runtime improvements in translation and
  synthesis.
\end{abstract}

\section{Introduction}

Deterministic $\omega$-automata are a key ingredient to many formal methods~\cite{vardi1985automatic}.
They also play a key role in
\emph{reactive synthesis}, which takes temporal specifications, typically in Linear Temporal Logic (LTL)~\cite{pnueli.77.focs}, and converts them to strategies that ensure their satisfaction \cite{pnueli.89.popl}.
However, the computational complexity and the complex algorithms that are involved in synthesis and, as a consequence, the capacity of tools that support synthesis techniques in their full generality, have proven to be a barrier for implementation~\cite{kupferman.focs.2005,finkbeiner.16.dsse}.
Two major successes in this field have been the syntactic fragment of LTL called GR[1] \cite{piterman.06.vmcai}, and
\LTLf~\cite{degiacomo.13.ijcai}, which is a semantic variant of LTL that considers finite rather than infinite computations.
In both cases, it is possible to use efficient techniques
for handling sets of states, leading to increased capacity.
Furthermore, algorithmic techniques for analysis of state spaces and translation of specifications to deterministic automata are simpler~\cite{BloemJPPS12,degiacomo.15.ijcai}.
In particular, in the context of \LTLf synthesis, tools leverage the algorithmic simplicity of deterministic finite automata~(DFA), in contrast to $\omega$-automata~\cite{zhu.17.ijcai,bansal.20.aaai,degiacomo.21.icaps,li.25.ecai}.

\emph{Obligation properties} \cite{manna.87.podc}, which can be expressed as a boolean combination of \emph{safety} and \emph{guarantee} properties, are good candidates for a fragment of LTL where we could hope to achieve efficient translation and synthesis, exploiting the simplicity of deterministic automata. Although limited in their expressivity, obligation properties are very important in practice. Specifications that are used in model checking or synthesis are often simple (safety, guarantee, and obligation).
\marginpar{Cf. App.~\ref{app:obligfreq}}
For instance, the 55 patterns of Dwyer et al.~\cite{dwyer.98.fmsp} contain 40 obligations. Similarly, Somenzi \& Bloem's compilation of 25 LTL formulas ``found in the literature''~\cite[Table~1]{somenzi.00.cav} contains 16 obligations.  Out of the 959 unique LTL specifications collected by the Synthesis Competition~\cite{jacobs.22.arxiv} at the time of writing, at least 262 specifications are obligations, so more than~27\%.
Efficient handling of obligation properties also benefits all cases where an obligation property appears as part of a Boolean combination with other temporal properties \cite{renkin.23.fmsd}.
Obligation properties are theoretically important, forming the second level in the safety-progress hierarchy of temporal properties~\cite{manna.87.podc,manna.92.book,manna.95.book,manna.10.response}.
Recently, the safety-progress hierarchy has attracted renewed interest thanks to a normalization procedure that reduces the nesting depth of strong and weak operators~\cite{esparza.24.acm}.

Safety and guarantee properties can be mapped to very simple \emph{deterministic} $\omega$-automata~\cite{cerna.03.mfcs}: a deterministic $\omega$-automaton for a safety property only rejects runs that can reach a rejecting sink, and a deterministic $\omega$-automaton for a guarantee only accepts runs that can reach an accepting sink. Boolean combinations of safety and guarantee can be mapped to \emph{deterministic weak automata} (DWA)~\cite{loding.98.msc,loding.01.ipl}, in which each \emph{strongly connected component} (SCC) contains only rejecting cycles or only accepting cycles.  The simple nature of DWAs makes them very easy to use: they are closed under Boolean operations using algorithms similar to those of DFAs, and, after a preprocessing stage with linear-cost~\cite{loding.01.ipl}, they can also be minimized like DFAs.

Unfortunately, testing whether an LTL formula is an obligation property is quite nontrivial in
general~\cite{dax.07.atva}.\marginpar{Cf. App.~\ref{app:classmembership}}
There are, however, well-known syntactic fragments of LTL that capture safety and guarantee properties.
By considering their Boolean combinations, we get a syntactic fragment of LTL which we shall refer to as \emph{syntactic obligation}~\cite{chang.92.icalp,cerna.03.mfcs}.
Importantly, every LTL formula representing an obligation property has an equivalent syntactic-obligation formula~\cite{chang.92.icalp}.
Furthermore, many of the previously mentioned examples are
syntactic-obligation formulas.

In this work, we use, on the one hand, the simple syntactic structure of syntactic obligations to easily identify obligation properties and, on the other hand, the DFA-like nature of DWAs to improve the translations and synthesis procedures for them. We generalize our recent work on a compact automaton representation of DFAs using Multi-Terminal Binary Decision Diagrams~(MTBDDs),  
that enabled more efficient \LTLf translation and synthesis~\cite{duret.25.ciaa}.
We show that syntactic obligations can be translated efficiently into DWAs represented using MTBDDs, and that the synthesis problem for syntactic obligations can be solved by interpreting the structure of
the MTBDDs as a \emph{weak Büchi game} that can be constructed on-the-fly.

We implemented this approach in Spot, a C++ library for LTL and $\omega$-automata algorithms~\cite{duret.22.cav}. Our empirical evaluation shows that it significantly improves the generation of deterministic automata for syntactic-obligation properties as well as the synthesis performance for this class of properties. It is worth noting that by implementing efficient obligation translation and synthesis in Spot, we improve the performance of \emph{all} tools building upon Spot whenever they deal with syntactic obligations.

\section{Succinct Representation of Deterministic Büchi Automata}

In this section, we introduce a representation of deterministic Büchi
automata (DBA) based on multi-terminal binary decision diagrams (MTBDDs).

\begin{figure}[tb]
  \def\conj{}
  \begin{tikzpicture}[mediumautomaton,xscale=0.8,baseline=(current bounding box.center)]
    \node[initial,accepting,state] (4) {$\varphi_1$};
    \node[accepting,state,right=of 4] (1) {$\varphi_2$};
    \node[state,below=of 4] (2) {$\varphi_4$};
    \node[state,right=of 2] (3) {$\varphi_3$};
    \node[state,right=of 1] (5) {$\varphi_5$};
    \node[accepting,state,right=of 3] (0) {$\ttrue$};
    \path[->] (4) edge[loop above] node[auto]{$i_1\conj o$} (4)
    (1) edge[loop above] node[auto]{$\bar{i_1}\conj i_2\conj o$} (1)
    (2) edge[loop below] node[auto]{$i_1$} (2)
    (3) edge[loop below] node[auto]{$\bar{i_1}\conj i_2$} (3)
    (5) edge[loop above] node[auto]{$o$} (5)
    (0) edge[loop below] node[auto]{$\top$} (0)
    (4) edge[bend right=20] node[above=-1pt]{$\bar{i_1}\conj o$} (1)
    (4) edge node[left=-1pt,pos=.6]{$i_1\conj\bar{o}$} (2)
    (4) edge node[above=-1pt,sloped,near end]{$\bar{i_1}\conj\bar{o}$} (3)
    (1) edge[bend right=20] node[above=-1pt]{$i_1\conj i_2\conj o$} (4)
    (1) edge node[right=-1pt,pos=.6]{$\bar{i_1}\conj i_2\conj \bar{o}$} (3)
    (1) edge node[above=-1pt,sloped,near end]{$i_1\conj i_2\conj \bar{o}$} (2)
    (1) edge node[above=-1pt,sloped,near end]{$\bar{i_2}\conj \bar{o}$} (0)
    (1) edge node[above=-1pt]{$\bar{i_2}\conj o$} (5)
    (2) edge[bend right=25] node[above=-1pt]{$\bar{i_1}$} (3)
    (3) edge[bend right=10] node[above=-1pt]{$i_1\conj i_2$} (2)
    (3) edge node[above]{$\bar{i_2}$} (0)
    (5) edge node[right]{$\bar{o}$} (0)
    ;
  \end{tikzpicture}
  \begin{tabular}{l}
    $\varphi_1 =\G(i_1{\lor}\X i_2)\liff \G o$ \\
    $\varphi_2 = (i_2{\land}\G(i_1{\lor}\X i_2))\liff \G o$ \\
    $\varphi_3 = \lnot(i_2{\land} \G(i_1{\lor} \X i_2))$ \\
    $\varphi_4 = \lnot\G(i_1{\lor} \X i_2)$ \\
    $\varphi_5 = \lnot\G o$ \\
  \end{tabular}
  \begin{tikzpicture}[initial text={},yscale=0.75,baseline=(current bounding box.center)]
    \node[root] at (0,0) (r2) {$\varphi_4$};
    \node[root] at (.7,0) (r3) {$\varphi_3$};
    \node[root, accepting] at (1.4,0) (r0) {$\ttrue$};
    \node[root, accepting, initial above] at (2.1,0) (r4) {$\varphi_1$};
    \node[root, accepting] at (2.8,0) (r1) {$\varphi_2$};
    \node[root] at (3.5,0) (r5) {$\varphi_5$};

    \node[inode] at (0,-1) (i11) {$i_1$};
    \node[inode] at (.7,-1) (i12) {$i_1$};
    \node[inode] at (2.1,-1) (i13) {$i_1$};
    \node[inode] at (2.8,-1) (i14) {$i_1$};

    \node[inode] at (0.4,-2) (i21) {$i_2$};
    \node[inode] at (1,-2) (i22) {$i_2$};
    \node[inode] at (2.5,-2) (i23) {$i_2$};
    \node[inode] at (3.1,-2) (i24) {$i_2$};

    \node[inode] at (1.8,-3) (o1) {$o$};
    \node[inode] at (2.5,-3) (o2) {$o$};
    \node[inode] at (3.5,-3) (o3) {$o$};

    \node[termn] at (0,-4) (t3) {$\varphi_3$};
    \node[termn] at (.7,-4) (t2) {$\varphi_4$};
    \node[terma] at (1.4,-4) (t0) {$\ttrue$};
    \node[terma] at (2.1,-4) (t1) {$\varphi_2$};
    \node[terma] at (2.8,-4) (t4) {$\varphi_1$};
    \node[termn] at (3.5,-4) (t5) {$\varphi_5$};

    \draw[rlink] (r2) -- (i11);
    \draw[rlink] (r3) -- (i12);
    \draw[rlink] (r4) -- (i13);
    \draw[rlink] (r1) -- (i14);
    \draw[rlink] (r5) -- (o3);
    \draw[rlink] (r0) to[out=-90,in=90] (t0.75);

    \draw[low] (i11) to[out=-110,in=90] (t3.120);
    \draw[high] (i11) to[out=-95,in=90] (t2.120);
    \draw[low] (i12) -- (i21);
    \draw[high] (i12) -- (i22);
    \draw[low] (i21) to[out=-70,in=90] (t0.130);
    \draw[high] (i21) to[out=-110,in=90] (t3.90);
    \draw[low] (i22)  to[out=-70,in=90]  (t0.105);
    \draw[high] (i22)  to[out=-110,in=90]  (t2.90);

    \draw[low] (i13) to[bend right=10] (o1);
    \draw[high] (i13) to[bend right=15] (o2);
    \draw[low] (o1) to[out=-170,in=90] (t3.60);
    \draw[high] (o1) to[out=-55,in=90] (t1);

    \draw[low] (i14) -- (i23);
    \draw[high] (i14) -- (i24);

    \draw[low] (i23) -- (o3);
    \draw[high] (i23) -- (o1);
    \draw[low] (i24) -- (o3);
    \draw[high] (i24) -- (o2);

    \draw[low] (o2) to[out=-135,in=90,looseness=1.6] (t2.60);
    \draw[high] (o2) to[out=-55,in=90] (t4);

    \draw[low] (o3) to[out=-135,in=90,looseness=1.1] (t0.50);
    \draw[high] (o3) -- (t5);
  \end{tikzpicture}
  \caption{A DBA for the formula $\varphi_1=\G(i_1 \lor \X i_2) \liff \G o$, and its MTBDD representation.\label{fig:example}}
\end{figure}
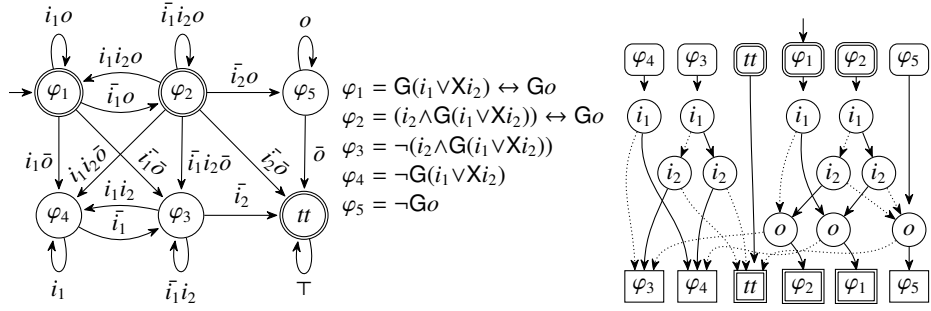

The left-hand side of Figure~\ref{fig:example} shows a deterministic Büchi
automaton over the set of atomic propositions $\CP=\{i_1,i_2,o\}$.
The \emph{letters} read by the automaton are assignments of all
propositions, so there are $2^{|\CP|}$ of them, but to simplify the
notations, we label the edges of the automaton by Boolean formulas (in
this example, implicit conjunctions) that are satisfied by all the
assignments that would normally label the edge.  In Spot, this
explicit representation of automata uses edges labeled by BDDs.  In
this example, we have also named the states of the automaton using LTL
formulas, but this is merely decorative.
This DBA has four strongly-connected components:
$\{\varphi_1,\varphi_2\}$, $\{\varphi_3,\varphi_4\}$, $\{\varphi_5\}$,
and $\{\ttrue\}$.  The automaton is \emph{weak} because each SCC contains
only accepting states, or only rejecting states.  We abbreviate
\emph{weak DBA} as \emph{DWA}.

The right-hand side of Figure~\ref{fig:example} shows a semi-symbolic
representation of the same DBA, where the outgoing transitions of each state are represented using MTBDDs.  The representation is inspired from the representation of DFAs in
Mona~\cite{henrikson.95.tacas,klarlund.01.tr}, and previous
work~\cite{duret.25.ciaa}.  The reader familiar with binary decision diagrams (BDDs)~\cite{bryant.86.tc} will recognize the classical structure of a forest of BDDs: \begin{tikzpicture}[baseline=(i.base)]
  \node[inode] at (0,0) (i) {$x$}; \node[overlay] at (0.75,0.15) (l)
  {$\ell$}; \node[overlay] at (0.75,-0.15) (h) {$h$};
  \draw[low,overlay] (i) -> (l); \draw[high] (i) -> (h);
\end{tikzpicture}~~~~should be interpreted as ``if $x$ then $h$ else
$\ell$'', atomic propositions appear in the same order on all paths,
and identical subtrees are shared.  Contrary to BDDs where leaves (or
terminals) are restricted to the values true or false,
MTBDDs~\cite{long.93.bdd,minato.96.vlsi,fujita.97.fmsd,klarlund.01.tr}
support arbitrary values: in our case, we encode the name of destination states
in the leaves.
Each root represents a function from Boolean assignments of atomic propositions to destination states, corresponding to transitions of the automaton.
A set of ``root pointers'' represented by nodes of the form
\begin{tikzpicture}[baseline=(i.base)]
  \node[root] at (0,0) (i) {$s$};
  \node[overlay] at (.7,0) (d) {};
  \draw[rlink] (i) -- (d);
\end{tikzpicture} point to an MTBDD representing the outgoing transitions of state~$s$.

Writing $\MTBDD(\CP,\CS)$ to denote an MTBDD with variables in
$\CP$ and terminals of type $\CS$, let us define the MTBDD-based
representation of a DBA (MTDBA for short).

\begin{definition}[MTDBA, MTDWA]\label{def:mtdba}\label{def:mtdwa}
  An MTDBA is a tuple $\CA=\langle \CQ, \CP, \iota, \Delta, \lambda \rangle$, where $\CQ$ is a finite set of \emph{states}, $\CP$ is a finite (and ordered) set of
  \emph{variables}, $\iota\in \CQ$ is the initial state,
  $\Delta: \CQ\to \MTBDD(\CP,\CQ)$ represents the outgoing edges of
  each state, and $\lambda: \CQ \to \BB$ represents the acceptance status of each state.
  If the MTDBA is weak, we may call it MTDWA.
\end{definition}

Each path from a root pointer to a terminal in the MTDBA
representation on the right-hand side of Figure~\ref{fig:example} corresponds to
an edge in the explicit DBA represented on the left-hand side.  The MTDBA
representation is more compact because it can share common subtrees.
For instance, if two states have exactly the same outgoing transitions,
but differ only by their acceptance, their root pointers will
designate the same MTBDD.

MTBDDs support unary and binary operations using algorithms similar to
those of BDDs.  For instance, given two $x,y\in\MTBDD(\CP,\CS)$, and a
binary operation $\odot:\CS\times\CS\to\CS$ over the terminals, one
can compute $z\in \MTBDD(\CP,\CS)$, such that for all assignments $v$
of all variables in $\CP$, the terminal $z(v)$ reached in $z$ by
following $v$ is equal to $x(v)\odot y(v)$.  We can therefore lift the
$\odot$ operation to the MTBDDs and write $z = x\odot y$.

\section{From Syntactic Obligations to MTDWA}\label{sec:translation}

This section shows how to build an MTDWA for a syntactic obligation.
Our construction is relatively straightforward, yet we could not find
any prior work for that fragment.

\subsection{Syntactic Sub-Classes of LTL}\label{sec:classes}

We assume that the reader is familiar with \LTL~\cite{pnueli.77.focs}.
The following grammar (where parentheses have been omitted, and $p\in\CP$)
defines four syntactic fragments~\cite{chang.92.icalp,cerna.03.mfcs}:\label{grammar-rules}
\begin{align*}
  \varphi_B ::={}& \ffalse\mid\ttrue\mid p\mid\lnot\varphi_B\mid\varphi_B\land\varphi_B
                   \mid\varphi_B\lor\varphi_B\mid\varphi_B\liff\varphi_B
                   \mid\varphi_B\lxor\varphi_B\mid\varphi_B\limplies\varphi_B
                   \mid\X\varphi_B \\
  \varphi_G ::={}& \varphi_B\mid \lnot\varphi_S\mid
                   \varphi_G\land \varphi_G\mid \varphi_G\lor \varphi_G
                   \mid\varphi_S\limplies\varphi_G\mid
                   \X\varphi_G \mid \F\varphi_G\mid
                   \varphi_G\U\varphi_G\mid \varphi_G\M\varphi_G \\
  \varphi_S ::={}& \varphi_B\mid \lnot\varphi_G\mid
                   \varphi_S\land \varphi_S\mid \varphi_S\lor \varphi_S
                   \mid\varphi_G\limplies\varphi_S\mid
                   \X\varphi_S \mid \G\varphi_S\mid
                   \varphi_S\R\varphi_S\mid \varphi_S\W\varphi_S \\
  \varphi_O ::={}& \varphi_G \mid \varphi_S\mid \lnot\varphi_O\mid
                   \varphi_O\land \varphi_O\mid \varphi_O\lor \varphi_O\mid
                   \varphi_O\liff \varphi_O\mid \varphi_O\lxor \varphi_O\mid
                   \varphi_O\limplies \varphi_O \mid \X\varphi_O \\
           \mid{}& \varphi_O\U\varphi_G\mid
                   \varphi_O\R\varphi_S\mid
                   \varphi_S\W\varphi_O\mid
                   \varphi_G\M\varphi_O
\end{align*}

Formulas built using the rule for $\varphi_S$ (resp. $\varphi_G$) are
syntactic safety (resp. syntactic guarantee) and also correspond to
the class $\Sigma_1$ (resp. $\Pi_1$) defined by Esparza et
al.~\cite{esparza.24.acm}.  The ``bottom'' class, built using the
$\varphi_B$ rule, is the intersection of these two syntactic classes; it differs from Esparza's $\Delta_0$ class in that it allows $\X$ operators.  If we ignore the last four options in $\varphi_O$, the $\varphi_O$ rule represents Boolean combinations of syntactic safety and syntactic guarantee, and differs from Esparza's $\Delta_1$ class in that it is also closed under the $\X$ operator.  The last four options in $\varphi_O$ allow to capture more obligations syntactically~\cite{chang.92.icalp,cerna.03.mfcs}.

Note that we never require formulas to be in \emph{negative normal form} (NNF).  Translations that require NNF suffer a blowup from the removal of operators $\liff$ and $\lxor$.  Instead, since DWA are closed under Boolean operations, we support these operators naturally.

In what follows, let $\LTL_O(\CP)$ designate the set of syntactic obligations that can be built over atomic propositions $\CP$ and similarly for $\LTL_B(\CP)$, $\LTL_G(\CP)$, and $\LTL_S(\CP)$.

\subsection{Translation to MTDWA}\label{sec:tr}

We are going to build an MTDWA
$\langle \CQ, \CP, \iota, \Delta, \lambda \rangle$, where states
$\CQ\subseteq \LTL_O(\CP)$ are identified by syntactic obligations, and $\iota\in\LTL_O(\CP)$ is the formula we want to translate.

\paragraph{Computing successors using MTBDDs.}
Assuming $\boxed{\alpha}$ represents a terminal labeled by
$\alpha\in\LTL_O(\CP)$, the MTBDD representation of the successors
of a state, i.e., the function $\tr: \LTL_O(\CP) \to \MTBDD(\CP,\LTL_O(\CP))$, is
defined as follows:

\begin{gather*}
  \begin{aligned}
    \tr(\ttrue)&=\boxed{\ttrue}\quad & \tr(\ffalse)&=\boxed{\ffalse} & \quad\tr(a)=\begin{tikzpicture}[baseline=(a.base)]
    \node[circle,draw,inner sep=2pt](a) at (0,0) {$a$};
    \node[inner sep=0pt,overlay](0) at (-.4,-.8) {$\boxed{\ffalse}$};
    \node[inner sep=0pt,overlay](1) at (.4,-.8) {$\boxed{\ttrue}$};
    \draw[low,overlay](a) -- (0);
    \draw[high,overlay](a) -- (1);
  \end{tikzpicture}\quad\text{for any $a\in\CP$}\\
    \tr(\X\varphi) &= \boxed{\varphi} & \tr(\lnot\varphi) &= \lnot\tr(\varphi) \\
  \end{aligned}\\[1ex]
  \begin{aligned}
   \tr(\alpha \odot \beta) &= \tr(\alpha)\odot \tr(\beta)\mathrlap{\quad\text{for any Boolean operator $\odot\in\{\land,\lor,\limplies,\liff,\lxor,...\}$}}\\[-2pt]
   \tr(\alpha \U \beta) &= \tr(\beta)\lor(\tr(\alpha)\land\boxed{\alpha \U \beta}) &
   \tr(\alpha \M \beta) &= \tr(\beta)\land(\tr(\alpha)\lor\boxed{\alpha \M \beta}) \\[-2pt]
   \tr(\alpha \W \beta) &= \tr(\beta)\lor(\tr(\alpha)\land\boxed{\alpha \W \beta}) &
   \tr(\alpha \R \beta) &= \tr(\beta)\land(\tr(\alpha)\lor\boxed{\alpha \R \beta}) \\[-2pt]
   \tr(\F \varphi) &= \tr(\varphi)\lor\boxed{\F\varphi} &
   \tr(\G \varphi) &= \tr(\varphi)\land\boxed{\G\varphi} \\[-2pt]
\end{aligned}
\end{gather*}

\paragraph{Deciding acceptance of states.}
In this simple LTL fragment, the acceptance of a state
$\varphi\in\LTL_O(\CP)$ can be defined by simply considering the
Boolean combination of safety and guarantee formulas   that $\varphi$ represents.
This can be done using the function $\lambda(\varphi)$ defined below,
which checks the top-level operator of each
maximal temporal subformula of $\varphi$ (strong operators $\F$, $\U$, $\M$ are
rejecting; weak operators $\G$, $\W$, $\R$ are accepting) and builds their Boolean
combinations.
\begin{gather*}
  \begin{aligned}
    \lambda(\ffalse)=
    \lambda(\alpha \U \beta) =
    \lambda(\alpha \M \beta) =
    \lambda(\F\varphi) &= \bot &
    \lambda(\ttrue)=
    \lambda(\alpha \W \beta) =
    \lambda(\alpha \R \beta) =
    \lambda(\G\varphi) &= \top &
  \end{aligned}\\
  \lambda(\alpha \odot \beta) = \lambda(\alpha)\odot\lambda(\beta)\quad\text{for any Boolean operator $\odot\in\{\land,\lor,\leftarrow,\leftrightarrow,\oplus,...\}$} \\
  \begin{aligned}
    \lambda(\lnot\varphi) &= \lnot\lambda(\varphi) &
    \lambda(\X\varphi) &= \lambda(\varphi) &
    \lambda(a) &= * ~\text{for any $a\in\CP$} \\
  \end{aligned}
\end{gather*}
Value $*$ is a wildcard that should be ignored during Boolean operations, as shown below.
\begin{align*}
  \top\land * &= \top   & \bot \land * &= \bot  & * \land * &= * & \top\oplus * &= \top & \bot \oplus * &= \top \\[-2pt]
  \top\lor *  &= \top   & \bot \lor * &= \bot  & * \lor *  &= * & * \oplus * &= *  & \lnot * &= * \\[-2pt]
  \top\rightarrow * &= \bot & \bot \rightarrow * &= \top & * \rightarrow * &= * &
  * \rightarrow \top &= \top & * \rightarrow \bot &= \bot & \\[-2pt]
  \top\leftrightarrow * &= \top & \bot \leftrightarrow * &= \top & * \leftrightarrow * &= *
\end{align*}
A state $\varphi$ such that $\lambda(\varphi)=*$ is a Boolean
combination of atomic propositions, possibly with $\X$ operators.
As it cannot be part of a cycle, its acceptance can be arbitrary.

In our implementation $\ttrue$ and $\ffalse$ never occur as arguments of
Boolean operators or of $\X$: such formulas are automatically
simplified (e.g., trying to construct $\X(\ttrue)\land a$ will return
$a$).  Therefore, $\lambda(\ttrue)$ and $\lambda(\ffalse)$ are never called
recursively.


\paragraph{Ensuring termination with propositional equivalence.}\label{sec:propeq}
It is tempting to define our automaton
$\langle \CQ, \CP, \iota, \Delta, \lambda \rangle$
by letting $\Delta=\tr$, and by setting $\CQ$ to be equal
 to all the formulas that can be reached transitively
from $\iota$ by applying $\tr$. This does not work because the set $\CQ$ constructed this way is not necessarily finite.  \marginpar{Cf. App.~\ref{app:notfinite}} The classical solution, which we reuse, is to use propositional equivalence~\cite{esparza.18.lics,duret.25.ciaa}.

In previous work on \LTLf translation~\cite{duret.25.ciaa}, we used
propositional equivalence to simplify formulas created at every step
of the computation of $\tr(\cdot)$.  In the current implementation we
found that it was faster to do it as a separate pass: first compute
the MTBDD $m=\tr(\varphi)$, then replace all the leaves of $m$ by a
representative of their propositional-equivalence class. This way, we
avoid computing propositional equivalence for the intermediate
computations of $\tr(\cdot)$.

Additionally, as part of propositional equivalence we consider the
distribution of the $\X$ operators through Boolean operators.  E.g., we interpret $\X(\alpha \land \beta)\lor \X(\beta)$ as $(\X(\alpha) \land \X(\beta))\lor \X(\beta)$, so it can be found to be equivalent to $\X(\beta)$.

From now on, we assume that $\tr$ includes propositional-equivalence,
so that $\CQ$, the set of formulas reachable transitively
from $\iota$ by applying $\tr$, is guaranteed to be finite.

\paragraph{Correctness of the construction.} While propositional
equivalence ensures that our construction terminates, its correctness
is established by the following lemma and theorem.

\marginpar{Cf. App.~\ref{app:proofoflemma}}%
\begin{restatable}{lemma}{structureofdwa}
\label{lem:structureofdwa}
  For any $\iota\in\LTL_O(\CP)$ let $D_\iota=\langle \CQ, \CP, \iota, \tr, \lambda \rangle$ be the automaton constructed as described above, keeping $\lambda:\CQ\to\BB\cup\{*\}$ (unlike in Definition~\ref{def:mtdba}).   The following statements hold:
  \begin{enumerate}
  \item If a state $\varphi\in\CQ$ is such that
    $\lambda(\varphi)=*$, this state belongs to a trivial SCC of $D_\iota$.
  \item For any bottom formula $\iota\in\LTL_B(\CP)$, the only possible non-trivial SCCs
    of $D_\iota$ are the states $\ttrue$ and $\ffalse$ (at least one of these has to exist).
  \item For any guarantee formula $\iota\in\LTL_G(\CP)$, all states $\varphi\in\CQ\setminus\{\ttrue\}$ that belong to non-trivial SCCs of $D_\iota$ have $\lambda(\varphi)=\bot$.
  \item For any safety formula $\iota\in\LTL_S(\CP)$, all states $\varphi\in\CQ\setminus\{\ffalse\}$ that belong to non-trivial SCCs of $D_\iota$  have $\lambda(\varphi)=\top$.
  \item For any obligation formula $\iota\in\LTL_O(\CP)$, all the states
    of any non-trivial SCC of $D_\iota$ get assigned the same value by $\lambda$, and this
    value is either $\top$ or $\bot$.
  \end{enumerate}
\end{restatable}

Since the states $\varphi$ for which $\lambda(\varphi)=*$ cannot
be part of a cycle, their acceptance can be chosen arbitrarily.
In Section~\ref{sec:mtminim}, we show that this choice is constrained
if we want to minimize the automaton.  For now, in order to formalize this definition, we declare
  those states as rejecting using the function $\lambda':\CQ\to\BB$
  defined as follows: \[
    \lambda'(\varphi)=\begin{cases}
      \top & \text{if}~\lambda(\varphi)=\top\\
             \bot & \text{if}~\lambda(\varphi)\in\{\bot,*\}
    \end{cases}
  \]

\begin{theorem}
    \label{cor:Correctness of MTDWA translation}
    For every $\iota\in\LTL_O(\CP)$ the automaton $D_\iota=\langle \CQ, \CP, \iota, \tr, \lambda' \rangle$ constructed above is an MTDWA, and it satisfies $L(\iota)=L(D_\iota)$.
\end{theorem}

\begin{proof}
    The fact that $D_\iota$ is weak is a corollary that follows directly from Lemma~\ref{lem:structureofdwa}.
    The language $L(D_\iota)$ being equivalent to $L(\iota)$ follows by noticing that $\tr$ implements the local truth of LTL formulas and by induction on the structure of formulas.\qed
\end{proof}

\pagebreak[3]
\subsection{Minimizing MTDWA}\label{sec:mtminim}

Contrary to DBAs~\cite{schewe.10.fsttcs}, DWAs can be minimized in
polynomial time~\cite{loding.01.ipl}.  In fact, after a very cheap
preprocessing to decide for each transient state (i.e., a state that
cannot be part of a loop) whether it should be accepting or rejecting,
a DWA can be minimized as if it were a DFA.
Löding's preprocessing~\cite{loding.01.ipl} is based on a simple ranking
function over the strongly connected components (SCC) of the DWA, such that the parity of the rank indicates the acceptance of the SCC.
We can perform an SCC decomposition over the MTDWA directly, by simply
interpreting the forest of MTBDDs (which is normally acyclic) as a
graph where a leaf $\boxed{\alpha}$ is connected to the root of
$\Delta(\alpha)$.

The minimization itself is easily done over MTBDDs by using a variant
of Moore's partition-refinement algorithm~\cite{moore.56.as} similar
to what is implemented in Mona~\cite{henrikson.95.tacas}.  Assuming that
each state labeled by $\alpha$ of the MTDWA is assigned to block
$B(\alpha)\in\NN$ in the partition of the current iteration, the next
partition can be found by creating for each state $s$ a temporary
$\Delta'(s)\in\MTBDD(\CP,\NN)$ in which all leaves $\boxed{\alpha}$ of
$\Delta(s)$ are replaced by $\boxed{B(\alpha)}$, and partitioning the
set of states according to the different MTBDDs in $\Delta'$.

An optimization not mentioned in Löding's paper is that the only
states that can be merged during minimization are those that share the
same rank.  Consequently, Löding's ranking function can be used
to define the initial partition of the minimization.

\section{$\LTL_O$ Reactive Synthesis}\label{sec:synthesis}

A reactive controller can be thought of as an electronic circuit that
produces output signals based on a history of input signals.  LTL
Reactive Synthesis is the problem of building such a controller
(usually as a Mealy machine) such that the combined histories of input
and output signals satisfy some given LTL specification.

We implement synthesis to AIGER
circuits~\cite{biere.11.tr,renkin.23.fmsd}, but for lack of space we
only discuss the ``realizability'' problem, defined as follows.
\begin{definition}[\cite{kupferman.focs.2005,finkbeiner.16.dsse}]\label{def:realizability}
  Given two disjoint sets of variables $\CI$ (inputs) and $\CO$
  (outputs), a controller is a function $\rho: (\BB^\CI)^*\to\BB^\CO$ that,
  given a history of assignments of input variables,
  produces an assignment of output variables.

  Given a word of input assignments $\sigma\in(\BB^\CI)^\omega$, the
  controller can be used to generate a word of output assignments
  $\sigma_\rho\in(\BB^\CO)^\omega$.
  The definition of $\sigma_\rho$ may
  use two semantics depending on whether we want the controller to have access to the current input assignment to decide the output assignment:
  \begin{description}[nosep]
  \item[Mealy semantics:] $\sigma_\rho(i)=\rho(\sigma(..i))$ for all $i\ge 0$.
  \item[Moore semantics:] $\sigma_\rho(i)=\rho(\sigma(..i-1))$ for all $i\ge 0$.
  \end{description}

  A formula $\varphi\in\LTL(\CI\cup\CO)$ is said to be \emph{Mealy-realizable}, or \emph{Moore-realizable}, if there exists a controller $\rho$ such that for every word $\sigma\in(\BB^\CI)^\omega$, it holds that $(\sigma\sqcup\sigma_\rho)\in\lang(\varphi)$ using the desired semantics.
\end{definition}
Formula $\varphi_1$ (from Figure~\ref{fig:example}) is both Mealy-realizable and
Moore-realizable.

For general LTL, synthesis or realizability can be reduced to the
translation of the specification into a deterministic parity
automaton, which is then interpreted and solved as a two-player game
with parity acceptance: one player plays the input signals,
the second player plays the output signals.  The specification is
realizable if the output player has a strategy to ensure that the
parity acceptance is satisfied.

For syntactic obligations, realizability and synthesis are a lot easier: any $\LTL_O$
formula can be converted to a DWA, which can also be interpreted as a
two-player game with \text{weak} acceptance.  Such games are known to
be solvable in linear time \cite[Section 6.1]{amram.21.cav}.  They can
also be reduced to the emptiness check of 1-letter weak alternating
automata~\cite[Theorem 4.7]{kupferman.00.acm}, or seen as a special
case of weak parity games~\cite{chatterjee.08.tr}.  These three
algorithms work similarly.  Since the automaton is weak, we know that
the cycles in one SCC are either all accepting (winning for the output
player) or all rejecting (winning for the input player), but except
for SCCs without successors, it might be possible for players to escape
an SCC that is losing for them.  The game is therefore solved by
enumerating the SCCs and processing them bottom-up: terminal SCCs can
be determined as winning for one player or the other according to
their acceptance.  Then the attractor of these states is grown by
backpropagation.  Moving up in the SCCs, undetermined states can be determined according to the acceptance of the SCC.

This algorithm can be implemented directly on the MTDWA
structure, by interpreting it as a game where the input (resp., output)
player makes the decision for the input (resp., output) variables, and where a
play that reaches $\boxed{\alpha}$ continues from the root node
associated with $\alpha$. For Mealy semantics, input variables should
be ordered to appear before output variables; for Moore semantics, it should be the
opposite.

We implement this game solving over MTDWA on-the-fly, together with
the automaton construction of Section~\ref{sec:translation}.
\marginpar{Cf. App.~\ref{app:synthesis-example}}
The setup is very similar to our previous solution for \LTLf
synthesis~\cite{duret.25.ciaa} where we had to construct a DFA
represented using MTBDDs, and interpret it as a reachability game
to be solved on-the-fly.   In a reachability
game, the exploration order is not important: this previous work used
a BFS just because we found it marginally better than DFS.  Here, for
obligation synthesis, we have to explore the states in a topological
order of their SCCs.
Enumerating SCCs during the on-the-fly construction of an automaton is a
well studied topic, and is typically used by on-the-fly explicit LTL
model
checkers~\cite{geldenhuys.05.tcs,gaiser.09.memics,renault.13.lpar}.
These algorithms are all based on a DFS exploration of the automaton.
Our implementation uses Dijkstra's SCC
algorithm~\cite{dijkstra.73.ewd376,dijkstra.76.dop}, and we construct
the MTDWA as the DFS progresses.  The implementation reuses the linear
backpropagation over incomplete graphs from previous
work~\cite[Algorithm 1]{duret.25.ciaa}.

\section{Implementation and Evaluation}\label{sec:evaluation}

The new constructions for translation, minimization, and game solving have been implemented~\cite{duret.26.zenodo}\marginpar{Cf. App.~\ref{app:artifact}} in Spot~\cite{duret.22.cav} and released as version 2.15.  Many graphical examples of these constructions can be found at
\url{https://spot.lre.epita.fr/ipynb/mtdswa.html}
Below, we evaluate how the translation and synthesis improved compared to the previous version.  We also include a comparison to a few third-party tools for reference.

\subsection{Translation}\label{sec:trans-benchmark}

\definecolor{color1}{HTML}{F8766D}
\definecolor{color2}{HTML}{B79F00}
\definecolor{color3}{HTML}{00BA38}
\definecolor{color4}{HTML}{00BFC4}
\definecolor{color5}{HTML}{619CFF}
\definecolor{color6}{HTML}{F564E3}
\def\pnewmin{\textcolor{color1}{\textsf{new-min}}}
\def\pnew{\textcolor{color3}{\textsf{new}}}
\def\poldmin{\textcolor{color2}{\textsf{old-min}}}
\def\pold{\textcolor{color4}{\textsf{old}}}
\def\ltltela{\textcolor{color6}{\textsf{ltl3tela}}}
\def\ltldela{\textcolor{color5}{\textsf{ltl2dela}}}

The new translation procedure is used by Spot's
\texttt{translator} class whenever the input formula is a
syntactic obligation and the user requests a deterministic output. For evaluation purposes, we use Spot's \texttt{ltl2tgba}
command-line tool, which converts LTL into various kinds of automata and outputs them in the HOA format~\cite{babiak.15.cav}.

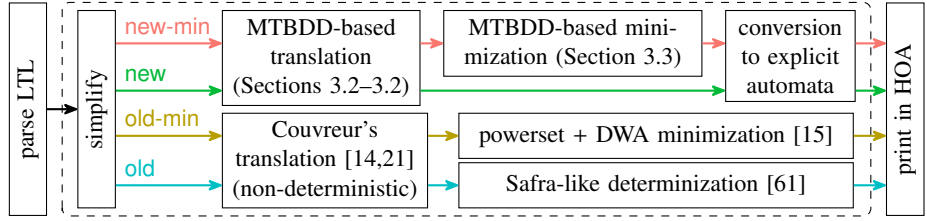
\begin{figure}[tb]
\begin{tikzpicture}[link/.style={thick,->}]
  \node[draw, text width=2.6cm, align=center, rotate=90] (parse) {parse LTL};
  \node[draw, text width=2.4cm, align=center, rotate=90, right=4mm of parse.south,anchor=north] (simplify) {simplify};
  \node[draw, align=center, text width=2.4cm,right=14mm of simplify.south east,anchor=north west,minimum height=3.5em] (mttrans) {MTBDD-based translation (Sections~\ref{sec:tr}--\ref{sec:propeq})};
  \node[draw, align=center, text width=2.5cm,right=14mm of simplify.south west,anchor=south west,minimum height=3.5em] (couvreur) {Couvreur's translation~\cite{couvreur.99.fm,duret.14.ijccbs} (non-deterministic)};
  \node[draw, align=center,text width=3.2cm,right=3mm of mttrans.north east,anchor=north west] (mtbddmin) {MTBDD-based minimization (Section~\ref{sec:mtminim})};
  \node[draw, align=center,text width=1.5cm,minimum height=3.5em,right=3mm of mtbddmin.north east,anchor=north west] (conv) {conversion to explicit automata};
  \node[draw, align=center,text width=5cm,right=4mm of couvreur.north east,anchor=north west] (wdbamin) {powerset + DWA minimization~\cite{dax.07.atva}\strut};
  \node[draw, align=center,text width=5cm,right=4mm of couvreur.south east,anchor=south west] (safra) {Safra-like determinization~\cite{redziejowski.12.fi}\strut};
  \coordinate (theend) at ($(parse -| conv.east)+(4mm,0)$);
  \node[draw, text width=2.6cm, align=center, rotate=90, anchor=north] at (theend) (print) {print in HOA};
  \draw[rounded corners=1mm,dashed] ($(simplify.north west)+(-2mm,-1mm)$) |-
  ($(conv.north east)+(2mm,1mm)$) |- cycle;
  \draw[link] (parse) -- (simplify);

  \draw[link,color=color1] (simplify.south |- mtbddmin) -- node[at start,above right]{\textsf{new-min}}
              (mttrans.west |- mtbddmin);
  \draw[link,color=color1] (mttrans.east |- mtbddmin) -- (mtbddmin);
  \draw[link,color=color1] (mtbddmin) -- (conv.west |- mtbddmin);
  \draw[link,color=color1] (conv.east |- mtbddmin) -- (print.north |- mtbddmin);

  \draw[link,color=color3] (simplify.south |- conv.-152) -- node[at start,above right]{\textsf{new}}
              (mttrans.west |- conv.-152);
  \draw[link,color=color3] (mttrans.east |- conv.-152) -- (conv.-152);
  \draw[link,color=color3] (conv.east |- conv.-152) -- (print.north |- conv.-152);

  \draw[link,color=color2] (simplify.south |- wdbamin) -- node[at start,above right]{\textsf{old-min}}
              (couvreur.west |- wdbamin);
  \draw[link,color=color2] (couvreur.east |- wdbamin) -- (wdbamin);
  \draw[link,color=color2] (wdbamin) -- (print.north |- wdbamin);

  \draw[link,color=color4] (simplify.south |- safra) -- node[at start,above right]{\textsf{old}}
              (couvreur.west |- safra);
  \draw[link,color=color4] (couvreur.east |- safra) -- (safra);
  \draw[link,color=color4] (safra) -- (print.north |- safra);
\end{tikzpicture}
\caption{Four pipelines usable by \texttt{ltl2tgba} to translate a syntactic obligation:
  \pnewmin, \pnew, \poldmin, and \pold.  The \texttt{spot::translator} class implements the dashed box.\label{fig:trans-pipeline}}
\end{figure}

Figure~\ref{fig:trans-pipeline} shows four translation pipelines that \texttt{ltl2tgba} can use to translate syntactic obligations into DWAs.  The dashed rectangle corresponds to the aforementioned \texttt{translator} class: it translates an LTL formula into an explicit $\omega$-automaton, but it can be configured to use different approaches.  The path \poldmin{} corresponds to the default behavior of Spot~2.14: after some simple syntactic simplifications, the formula is converted into a non-deterministic Büchi automaton using Couvreur's translation algorithm~\cite{couvreur.99.fm,duret.14.ijccbs}, to obtain a weak NBA; the result is then determinized by powerset construction \marginpar{Cf. App.~\ref{app:classmembership}.} and minimized as described by Dax et al.~\cite{dax.07.atva}, and finally the resulting automaton is output in HOA.  The \pold{} pipeline is used when the DWA-minimization is explicitly disabled or not applicable: this path works for any input formula, not just obligations.  In that case, since the output of the translation is non-deterministic, but the goal is to have a deterministic automaton, the automaton is determinized using a Safra-based construction~\cite{redziejowski.12.fi}.\looseness=-1

In both \poldmin{} and \pold{}, automata are represented as
explicit graphs with edges labeled by Boolean formulas over
propositions (represented as BDDs)~\cite{duret.14.ijccbs,duret.22.cav}, as on the left-hand side of Figure~\ref{fig:example}.  In the new pipelines, called \pnewmin{} and \pnew{}, the MTBDD-based translation described in Section~\ref{sec:translation} produces an MTDWA as on the right-hand side of Figure~\ref{fig:example}, and that can optionally be minimized using the same data structure.  In order to output an MTDWA in the HOA format, we first convert this MTBDD-based representation into an explicit representation.  This conversion has a significant cost because it has to enumerate all paths in the MTBDDs.  Nonetheless, we will see that these new pipelines are still competitive, despite the costly conversion.

As far as we know, Spot is the only tool that can translate syntactic obligations into DWA that are guaranteed to be minimal (via the \poldmin{} or \pnewmin{} pipelines).  We can, however, compare to tools that produce deterministic $\omega$-automata even if they do not guarantee minimality.  We consider \ltltela~2.1.1~\cite{major.19.atva} and \ltldela{} from Owl 21.0~\cite{kretinsky.18.atva}.  (The latter is an evolution of \textsf{Delag}~\cite{muller.17.gandalf}, which we did not run.)

For a fair comparison, we set up all tools to read an LTL formula and produce an equivalent deterministic $\omega$-automaton (with any acceptance condition) in the HOA format.
We request complete automata from all tools except \ltltela{} (which lacks this option).
We measure the entire execution of each tool using
\href{https://spot.lre.epita.fr/ltlcross.html}{\texttt{ltlcross}} with a timeout of 60 seconds.  The experiment was run one task at a time on a dedicated Core i7-3770 with \emph{Turbo Boost} disabled, and the frequency scaled down to 1.6GHz to prevent CPU throttling.

We evaluated the translation pipelines on 494 syntactic obligations: 
\marginpar{Details in App.~\ref{app:trans-benchmark}}
\begin{itemize}[topsep=0pt]
\item 310 formulas representing instances of
  17 scalable patterns~\cite{geldenhuys.06.spin,cichon.09.depcos,geldenhuys.06.spin,kupferman.11.mochart,tabakov.10.rv};
\item 76 formulas collected from various works~\cite{dwyer.98.fmsp,somenzi.00.cav,etessami.00.concur,pelanek.07.spin,holevek.04.tr};
\item 108 unique formulas from the synthesis
  competition~\cite{jacobs.22.arxiv}: all the instances that are syntactic
  obligations, except for two formulas syntactically equivalent to
  $\ttrue$ or $\ffalse$.
\end{itemize}
Figure~\ref{fig:cactus-trans}, where the time axis is logarithmic, clearly shows how Spot's translation has improved over this set of formulas.  Note that \pnew{} and \pnewmin{} are barely distinguishable, suggesting that
the guarantee to build minimal automata comes at no extra cost.

\begin{figure}[tb]
  \includegraphics[width=\textwidth,alt={Cactus plot showing new and new-min configurations scaling better than other tools and older configurations}]{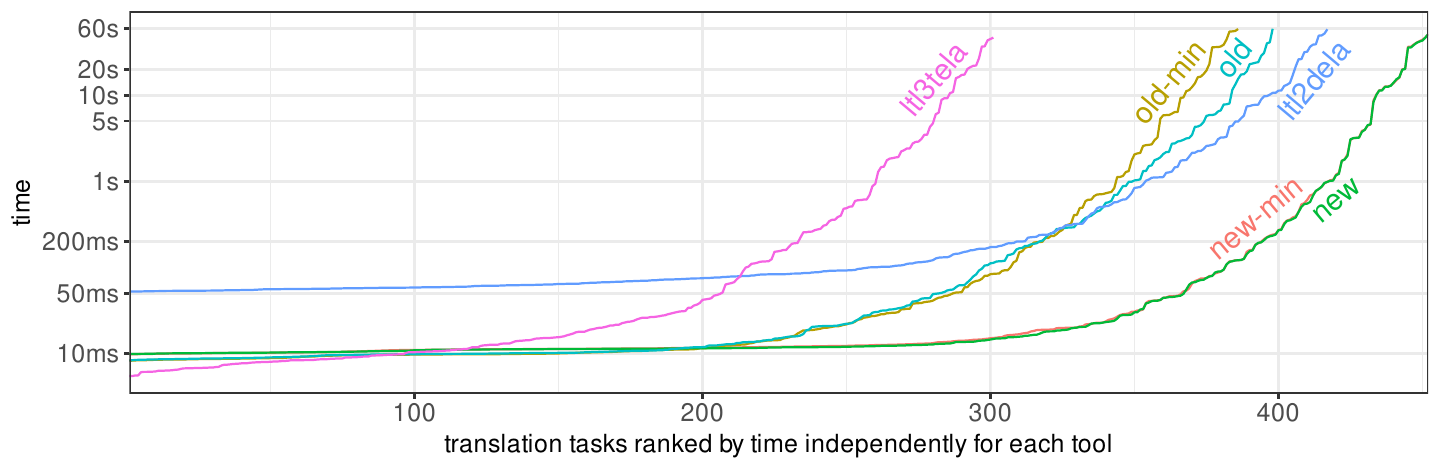}
  \caption{Comparison of translation tools on 494 syntactic obligations from the literature. \label{fig:cactus-trans}}
\end{figure}

Keep in mind, however, that we measured the runtime of a complete translation process, including the conversion to explicit automata (for \pnew{} and \pnewmin) and HOA output (for all tools).  We now turn to the evaluation of how the MTBDD-based representation of DWA improves synthesis: in this case, we do not need to convert automata to an explicit representation, saving a lot of time.

\subsection{Synthesis}\label{sec:bench-synthesis}

\def\newtrans{\textcolor{color1}{\textsf{(new trans.)}}}
\def\newsynt{\textcolor{color3}{\textsf{(new synt.)}}}
\def\oldtrans{\textcolor{color2}{\textsf{(old trans.)}}}
\def\strix{\textcolor{color4}{\textsf{Strix}}}
\def\semml{\textcolor{color6}{\textsf{SemML}}}
\def\ssyft{\textcolor{color5}{\textsf{SSyft}}}

Previous versions of Spot's \texttt{ltlsynt} converted an LTL \marginpar{Cf. App.~\ref{app:ltlsyntarch}}specification by translating it to a deterministic parity automaton (DPA), turning this DPA into a two-player game, and solving the game~\cite{renkin.23.fmsd}.  In the case of syntactic obligations, the DPA was obtained by the \poldmin{} translation pipeline in Figure~\ref{fig:trans-pipeline}.  This default behavior of \texttt{ltlsynt}~2.14 is what is called \oldtrans{} in this section.

The improvements described in this paper allow us to define two better configurations of \texttt{ltlsynt}.  \newtrans{} designates a configuration where the \texttt{spot::trans\-la\-tor} class, used to produce the DPA, is changed to use the \pnewmin{} translation pipeline in Figure~\ref{fig:trans-pipeline}, where the automaton is still converted to an explicit representation, to be solved as a game.  The \newsynt{} configuration corresponds to the new
default of \texttt{ltlsynt}, where syntactic obligations are converted into an MTBDD-based representation, and the game is solved directly on that representation as described in this paper.

\begin{figure}[tb]
  \includegraphics[width=\textwidth,alt={Cactus plot showing new synthesis and new translation scaling better than other tools and older configurations}]{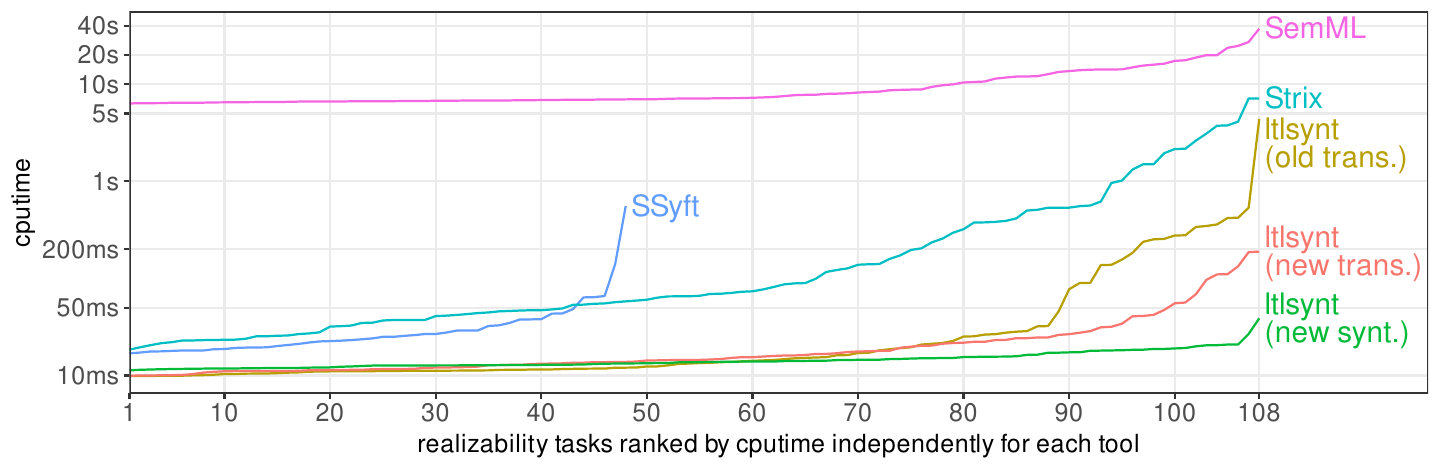}
  \caption{Comparison of the three configurations of \texttt{ltlsynt} over the 108 specifications from SyntComp that are syntactic obligations.  Third-party tools \strix{}, \semml{}, and \ssyft{}
  are included for reference.\label{fig:cactus-time}}
\end{figure}

Figure~\ref{fig:cactus-time} shows the improved runtimes as a cactus plot.  For this experiment, we used BenchExec~3.22~\cite{beyer.19.sttt} to track time and memory usage.
\marginpar{See App.~\ref{app:mem} for\\memory usage.}
We include \strix~21.0~\cite{luttenberger.19.ai} (the winner of
SyntComp'23) and \semml~\cite{kretinsky.25.tacas} (the winner of SyntComp'24) for comparison.  (The winner of SyntComp'25 was~\texttt{ltlsynt} due to a technical defect in \semml.)\marginpar{See App.~\ref{app:semml}\\about \semml.}
Moreover,
\href{https://github.com/Shufang-Zhu/Syft-safety-compositional}{\ssyft}~\cite{suguman.23.vstte},
patched to use the Mealy semantics, is included as an example of a tool that is dedicated to syntactic-safety formulas: there are 48 of those in this benchmark.  All tools were configured to check the specification for realizability: this way we do not include any time spent generating and outputting controllers.

\section{Conclusion}

We have shown that the MTBDD-based translation has made
\texttt{ltl2tgba} faster at converting syntactic obligations to (weak)
DBAs.  This holds even when we switch back to an explicit
representation.  In the case of synthesis, the latter switch back is
not required and the weak game can be solved directly on the
MTBDD-based representation, providing an even greater speedup.

The techniques described target syntactic obligations, so we focused
our experiments on those.  We expect, however, that those improvements
will impact other types of specifications as well.  Firstly, when a
specification can be seen as a conjunction of output-disjoint
sub-specifications, \texttt{ltlsynt} solves each sub-specification
independently~\cite{finkbeiner.21.nfm,renkin.23.fmsd}.\marginpar{See
  App.~\ref{app:decomp}.}\label{sec:decompmention} Therefore, even if
the full specification is not a syntactic obligation, it could contain
a sub-specification that is a syntactic obligation, and that will be
solved using the improved MTBDD-based technique.  Secondly, a possible
way to obtain a DPA is to first translate the specification into a
deterministic Emerson-Lei automaton (DELA), which can then be
converted into a DPA~\cite{renkin.20.atva,casares.22.tacas}.  To
obtain a DELA, Spot generalizes the compositional translation used by
the \texttt{Delag} tool~\cite{muller.17.gandalf}: the formula to
translate is broken into top-level Boolean operators, subformulas are
translated to deterministic automata according to their nature, and
those automata are then recombined.  In this approach, subformulas
that are syntactic obligations will now be translated much more
efficiently.

Obvious future work would be to generalize the MTBDD-based
translation, so that all LTL formulas can be turned into MTBDD-based
DELA, probably using $\Delta_2$-normalization~\cite{esparza.24.acm},
converting top-level temporal formulas to Büchi and co-Büchi
automata, and using Boolean combinations for the final automaton.

\begin{credits}
\subsubsection{Data Availability Statement}
%
Code and benchmark data are archived on Zenodo~\cite{duret.26.zenodo}.

\subsubsection{\ackname}
NP is supported by Swedish Research Council (VR) project (No. 2020-04963) and the Wallenberg AI, Autonomous Systems and Software Program (WASP) funded by the Knut and Alice Wallenberg Foundation. GDG is partially supported by the EPSRC grant EP/Y028872/1, Mathematical Foundations of Intelligence: An Erlangen Programme for AI.

\subsubsection{\discintname}
%
The authors have no competing interests to declare that are relevant
to the content of this article.
\end{credits}

\bibliographystyle{splncs04}
\bibliography{biblio}

\ifappendix
\newpage
\appendix

These appendices and the margin notes that point to them are for the
interested reviewers, but are not meant to be part of the final
version of the article.

\section{Obligation Frequency}\label{app:obligfreq}

\begin{table}[tb]
  \centering
\begin{tabular}{lccccccccccccc}
  \toprule
                                             &               & ~~~~~ & \multicolumn{5}{c}{property classes} & ~~~~~                      & \multicolumn{5}{c}{syntactic classes}                                                                                                      \\
  source                                     & set size      &        & $O$                                  & $O{\setminus}S{\setminus}G$ & $S{\setminus}B$   & $G{\setminus}B$  & $B$      &  & $O$           & $O{\setminus}S{\setminus}G$ & $S{\setminus}B$ & $G{\setminus}B$ & $B$ \\
  \cmidrule{1-2}
  \cmidrule{4-8}
  \cmidrule{10-14}
SyntComp~\cite{jacobs.22.arxiv}              & 959           &        & $\ge$262                             & $\ge$172                    & ${\ge}$73         & ${\ge}$7         & ${\ge}$7 &  & 110           & 59                          & 48              & 1               & 2   \\
Dwyer et al.~\cite{dwyer.98.fmsp}            & \phantom{0}55 &        & \phantom{$\ge$0}40                   & \phantom{$\ge$00}2          & \phantom{$\ge$}37 & \phantom{$\ge$}1 &          &  & \phantom{0}25 & 13                          & 11              & 1               &     \\
Somenzi\&Bloem~\cite{somenzi.00.cav}         & \phantom{0}27 &        & \phantom{$\ge$0}16                   & \phantom{$\ge$00}2          & \phantom{$\ge$0}6 & \phantom{$\ge$}6 & \phantom{$\ge$}2        &  & \phantom{0}13 & \phantom{0}4                & \phantom{0}4    & 5               &     \\
Etessami\&Holzmann~\cite{etessami.00.concur} & \phantom{0}12 &        & \phantom{$\ge$00}5                   &                             &                   & \phantom{$\ge$}4 & \phantom{$\ge$}1        &  & \phantom{00}5 &                             &                 & 5               &     \\
BEEM~\cite{pelanek.07.spin}                  & \phantom{0}20 &        & \phantom{$\ge$0}10                   & \phantom{$\ge$00}1          & \phantom{$\ge$0}6 & \phantom{$\ge$}1 & \phantom{$\ge$}2        &  & \phantom{00}7 & \phantom{0}2                & \phantom{0}4    & 1               &     \\
  Liberouter~\cite{holevek.04.tr}            & \phantom{0}55 &        & \phantom{$\ge$0}30                   &                             & \phantom{$\ge$}27 & \phantom{$\ge$}1 & \phantom{$\ge$}2        &  & \phantom{0}26 & \phantom{0}1                & 23              & 2               &     \\
  \bottomrule
\end{tabular}
\caption{\strut Various collections of LTL formulas classified in the lower-fragments of the temporal hierarchy, and in their syntactic fragments.  Empty cells are 0.\label{tab:obligfreq}}
\end{table}
Table~\ref{tab:obligfreq} shows the frequency of obligation formulas
in various sets of LTL formulas.  For each set, the column ``size'' is
the number of formulas in the set, and the remaining columns show how
many formulas fall into the various lower classes of the temporal
hierarchy, or in their syntactical fragments.

The letters $O$, $S$, $G$, $B$ refer respectively to obligations,
safety, guarantee, and bottom classes, where ``bottom'' is the
intersection of $S$ and $G$.  Since $O$ includes $S$ and $G$ and both
include $B$, we report instead on the difference between these sets.
E.g. $O{\setminus}S{\setminus}G$ counts obligation properties that are
neither safety nor guarantee properties.

The ``property class'' columns were obtained by running the
obligation detection algorithm presented in
Appendix~\ref{app:classmembership}.  This classification is incomplete
on the formulas from the Synthesis Competition because some of these
formulas are too big to be processed by Spot's classifier.  The lower
bounds provided were computed with a timeout of 60 seconds.  Out of
the 959 formulas, 262 (27.3\%) were classified as obligations, 401
(41.8\%) as not obligation (high classes in the hierarchy), and the
remaining 296 (30.9\%) formulas could not be classified either because
the classification would take more than 60 seconds, or because of
other errors from Spot (e.g., a translation requiring more than 32
acceptance sets).

The ``syntactic class'' columns were obtained by checking whether the
syntax tree of each formula follow the grammar rules from
Section~\ref{grammar-rules}.

For instance, in the 27 formulas collected by
Somenzi\&Bloem~\cite{somenzi.00.cav}, 16 represent obligation
properties.  In these 16 properties, 2 are obligations that are neither
safety nor guarantee, 6 are safety properties that are not guarantee
properties, 6 are guarantee properties that are not safety properties, and 2 are
both guarantee and safety.\footnote{The two formulas from
  Somenzi\&Bloem~\cite{somenzi.00.cav} that are in $B$ are:
  ``$(\X p_0 \U \X p_1) \lor \lnot\X(p_0 \U p_1)$'' and
  ``$(\X p_0 \U p_1) \lor \lnot \X(p_0 \U (p_0 \land p_1))$''.  They belong to $B$ because they are both tautologies.  Syntactically, however, the lowest fragment they belong to is $O$.}  If we consider only the syntactic characterization of these classes, 13 formulas are syntactic obligations, 4 are syntactic obligations that are neither syntactic-safety nor syntactic-guarantee formulas, 4 are syntactic-safety formulas that are not syntactic-guarantee formulas, 4 are syntactic-guarantees that are not syntactic-safety, and 0 formulas are both syntactic-safety and syntactic-guarantee.

We can see in that table that obligations represent a frequently
used class of formulas, and that, although the syntactic obligations
do not capture all obligations, they represent a large portion of those.

\section{Deciding Membership to the Classes of the Temporal Hierarchy}
\label{app:classmembership}

To play and get familiar with Manna \& Pnueli's hierarchy, the reader is
invited to go to \url{https://spot.lre.epita.fr/app/}, click on the
``\textsc{study}'' tab, and enter any LTL formula.   That should
display the class this formula belongs to, among other things.

For instance, the formula ``\verb+(a xor b) R (a U b)+'' will be
detected as an obligation property even if it is not a
syntactic obligation.  One equivalent syntactic obligation is the
formula ``\verb+(a xor b) M (a U b) | G(a & b)+''.

The classification performed on that page is done by Spot.  To
detect obligation properties, Spot implements a variation on the
technique of Dax et al.~\cite[Section 3.2]{dax.07.atva}.  Our
implementation improves upon the mentioned technique in the way it
finds accepting SCCs, and by using only a single inclusion check at
the end.  Given some input formula $\varphi$ for which we would like
to know if it is an obligation, Spot proceeds as follows:
\begin{enumerate}
\item translate $\varphi$ into a non-deterministic Büchi automaton
  $N_\varphi$,
\item ignoring the acceptance of states, use the powerset construction on
  $N_\varphi$ to obtain the deterministic \emph{structure} $D_\varphi$,
\item any SCC of $D_\varphi$ that intersects an accepting SCC of
  $N_\varphi$ in the synchronous product of
  $N_\varphi\otimes D_\varphi$ should be marked in $D_\varphi$ as
  fully accepting
\item at this point $D_\varphi$ is a DWA such that $\lang(\varphi) \subseteq \lang(D_\varphi)$ (because its accepting SCCs possibly capture
  runs that were not accepted by $N_\varphi$).
\item test $\lang(D_\varphi)\subseteq\lang(\varphi)$ by constructing a
  non-deterministic Büchi automaton $N_{\lnot\varphi}$ for the
  negation of $\varphi$ and then ensuring that the product
  $N_{\lnot\varphi}\otimes D_\varphi$ is empty.
\end{enumerate}
The success of the last inclusion check implies that
$\lang(D_\varphi)=\lang(\varphi)$.  Since $D_\varphi$ is a DWA, this
in turn implies that $\varphi$ is an obligation.  After minimizing this
DWA, we can also detect guarantee properties (the only accepting state
is an accepting sink) or safety properties (the only rejecting state is
a rejecting sink).

As the size of $N_\varphi$ could be exponential in $|\varphi|$, the size of $D_\varphi$ is, in the worst case, doubly exponential.
It follows that the algorithm above can be implemented in exponential space if $D_\varphi$ is constructed on the fly.

Detection of higher classes in the hierarchy uses other techniques that
are out of the scope of this paper.

\section{Why Propositional Equivalence is Needed}\label{app:notfinite}

Consider $\varphi_1=(\G a)\W (\G b)$.
Following the branch $a\land b$ in the MTBDD $\tr(\varphi_1)$ leads to a leaf labeled
by $\varphi_2=(\G b) \lor ((\G a) \land ((\G a)\W (\G b)))$.
Following $a\land b$ in the MTBDD $\tr(\varphi_2)$ leads to a leaf labeled
by $\varphi_3=(\G b) \lor ((\G a) \land ((\G b) \lor ((\G a) \land ((\G a)\W (\G b)))))$, etc.
Those formulas will keep growing.

Propositional equivalence can be defined as follows:

\begin{definition}[Propositional Equivalence~\cite{esparza.18.lics}]\label{def:propequiv}
  For $\varphi\in\LTL(\CP)$, let $\varphi_P$ be the
  Boolean formula obtained from $\varphi$ by replacing every
  maximal temporal subformula $\psi$ by a Boolean variable $x_\psi$.
  Two formulas $\alpha,\beta\in\LTL(\CP)$ are
  \emph{propositionally equivalent}, denoted $\alpha\equiv \beta$, if
  $\alpha_P$ and $\beta_P$ are equivalent Boolean formulas.
\end{definition}

If we compute the Boolean formulas $\varphi_{1p}$, $\varphi_{2p}$,
$\varphi_{3p}$ obtained by replacing $\G a$, $\G b$, and
$(\G a)\W (\G b)$ by three different Boolean variables, and if we
represent $\varphi_{1p}$, $\varphi_{2p}$, $\varphi_{3p}$ as BDDs, we
can see right away that $\varphi_{2p}=\varphi_{3p}$.  Therefore
$\varphi_{2}$ and $\varphi_{3}$ are propositionally equivalent, and we
can replace the latter by the former, keeping the set of reachable
formulas finite in the translation.

\section{Proof of Lemma~\ref{lem:structureofdwa}}\label{app:proofoflemma}

\structureofdwa*

\begin{proof}~\\[-1em]
    \begin{enumerate}
    \item[1.] By definition of $\lambda$, a state $\varphi\in\CQ$ such
      that $\lambda(\varphi)=*$ is necessarily a Boolean combination
      of atomic propositions, with the possible addition of $\X$
      operators.  By the definition of $\tr$, if $\varphi$ is a purely
      Boolean combination of atomic propositions (i.e., no $\X$), its
      successors can only be $\ttrue$ or $\ffalse$.  Additionally, if
      $\varphi$ contains $n$ nested $\X$, its successors will have at
      most $n-1$ nested $\X$.  Therefore $\varphi$ cannot be part of a
      cycle.
    \end{enumerate}

    \noindent
    The remaining points are shown by induction on the structure of
    the formula, following the fragments defined in Section~\ref{sec:classes}.
    \begin{enumerate}
\item[2.]
%
      Follows immediately from the previous point.  The only formulas
      of $\LTL_B(\CP)$ for which $\lambda$ does not return $*$ are
      $\ffalse$ and $\ttrue$, and the automaton constructed for
      $\iota\in \LTL_B(\CP)$ will necessarily reach one (or both) of
      those states.
    \item[3.]
%
        According to the definition of $\tr(\cdot)$, the only way to have a non-trivial SCC is by using the operators $\U$ and $\M$. It follows that in every state $\psi$ that appears in a non-trivial SCC we have $\lambda(\psi)=\bot$. Otherwise, the only accepting non-trivial SCC is $\ttrue$, for which $\lambda(\ttrue)=\top$.

        In the case of $\iota=\neg \varphi_S$ for $\varphi_S\in \LTL_S(\CP)$, the Lemma follows by duality and induction on the structure of the formula.
    \item[4.]
%
      This is the dual of the previous point.
    \item[5.]
        For all formulas built by using Boolean connectives or $\X$, this is proven by induction on the structure of the formula and the closure of weak automata under all Boolean operations.

        The remaining cases are formulas of the form $\varphi_O \U \varphi_G$, $\varphi_O \R \varphi_S$, $\varphi_S \W \varphi_O$, and $\varphi_G \M \varphi_O$.
        We concentrate on formulas of the form $\varphi_O\U\varphi_G$. The others are similar.

        Consider a formula $\psi^o \U\psi^g$. We assume by induction that the automata obtained for $\psi^o$ and $\psi^g$ satisfy the Lemma.
        Consider the automaton created for $\psi^o \U \psi^g$.
        By structure of $\tr(\cdot)$, a run of $D_\iota$ has the following form:
        $$
        \begin{array}{l l}
            q_0: & \psi^o \U \psi^g\\
            q_1:& \left (\psi^o \U\psi^g \wedge \psi^o_{0,1}\right ) \vee \psi^g_{0,1} \\
            q_2:&\left (\psi^o \U\psi^g \wedge \psi^o_{1,2} \wedge \psi^o_{0,2}\right ) \vee \left (\psi^g_{1,2} \wedge \psi^o_{0,2}\right ) \vee \psi^g_{0,2}\\
            q_3:&\left (\psi^o \U\psi^g \wedge \psi^o_{2,3} \wedge \psi^o_{1,3} \wedge \psi^o_{0,3}\right ) \vee \left (\psi^g_{2,3} \wedge \psi^o_{1,3} \wedge \psi^o_{0,3}\right ) \vee \left (\psi^g_{1,3} \wedge \psi^o_{0,3}\right ) \vee \psi^g_{0,3}\\
            \ldots\\
            q_i:& \left (\psi^o \U\psi^g \wedge \displaystyle \bigwedge_{0\leq j< i} \psi^o_{j,i}\right ) \vee \bigvee_{0\leq j<i} \left (\psi^g_{j,i} \wedge \displaystyle \bigwedge_{1\leq k<j} \psi^o_{k,i}\right ),
            \end{array}
        $$
        where $\psi^\alpha_{j,k}$ represents the state of the automaton for $\psi^\alpha$ after having read the input $w_{j,k}$. Notice that formulas are simplified by the propositional equivalence. In particular, $\psi^\alpha_{j,k}$ could be equivalent to $\ttrue$ or $\ffalse$, leading to a much simpler structure.

        Thus, the general form of a state of $D_\iota$ is (propositionally equivalent to):
        $$
          \left (\psi^o \U\psi^g \wedge \displaystyle \bigwedge_{0\leq j< i} \psi^o_{j,i}\right ) \vee \bigvee_{0\leq j<i} \left (\psi^g_{j,i} \wedge \displaystyle \bigwedge_{1\leq k<j} \psi^o_{k,i}\right ).
        $$

        We notice that for every state $\varphi$ that has $\psi^o \U\psi^g$ as a component in the first disjunct, we have $\lambda(\varphi) \in \{\bot,\top\}$. Consider a state $\varphi$ where the component containing $\psi^o\U\psi^g$ has been simplified to $\ffalse$ or $\ttrue$.
        If state $\varphi$ appears in a non-trivial SCC, then $\lambda(\varphi)\in \{\bot,\top\}$ as the same holds for the states of $D_{\psi^o}$ and $D_{\psi^g}$.

        We have to show that every non-trivial SCC is assigned a uniform acceptance status of either $\bot$ or $\top$.
        Consider a non-trivial SCC where $\lambda$ associates with some states $\bot$ and some states $\top$.
        Consider a loop that starts with $\lambda$ being $\bot$, then $\top$, then $\bot$ again.
        The only states with which $\lambda$ associates $\top$ are those for which $\lambda(\psi^g_{j,k})=\top$. However, by induction, this
        is only the case if $\psi^g_{j,k}\equiv \ttrue$. This is irreversible and hence the disjunct to which this $\psi^g_{j,k}$ belongs is a weak automaton that is obtained as the product of a fixed number of weak automata.
        By the smaller weak automata having every non-trivial SCC with uniform acceptance, it is impossible to have such an SCC with mixed acceptance.
        Hence, we conclude that the only way to go back from $\lambda$ associating $\top$ to $\bot$ is if some conjunct containing $\psi^g_{j,k}\equiv \ttrue$ simplifies overall to $\ffalse$.
        That is, $\bigwedge_{1\leq k <j}\psi^o_{k,i} \equiv \ffalse$.
        We notice that the first disjunct, which contains $\psi^o\U \psi^g$, includes $\bigwedge_{1\leq k <j}\psi^o_{k,i}$ as a conjunct.
        Hence, the first disjunct, overall, becomes equivalent to $\ffalse$ as well.
        It follows that every non-trivial SCC with both accepting and rejecting states is part of the product of a fixed number of weak automata. By weak automata being closed under Boolean operations this is impossible.
    \end{enumerate}
\qed
\end{proof}

\section{Synthesis Example}\label{app:synthesis-example}

\begin{figure}[tb]
  \centering
  \begin{tikzpicture}[initial text={},yscale=0.75,baseline=(current bounding box.center)]
    \node[root] at (0,0) (r2) {$\varphi_4$};
    \node[root] at (.7,0) (r3) {$\varphi_3$};
    \node[root, accepting] at (1.4,0) (r0) {$\ttrue$};
    \node[root, accepting, initial above] at (2.1,0) (r4) {$\varphi_1$};
    \node[root, accepting] at (2.8,0) (r1) {$\varphi_2$};
    \node[root] at (3.5,0) (r5) {$\varphi_5$};

    \node[inode] at (0,-1) (i11) {$i_1$};
    \node[inode] at (.7,-1) (i12) {$i_1$};
    \node[inode] at (2.1,-1) (i13) {$i_1$};
    \node[inode] at (2.8,-1) (i14) {$i_1$};

    \node[inode] at (0.4,-2) (i21) {$i_2$};
    \node[inode] at (1,-2) (i22) {$i_2$};
    \node[inode] at (2.5,-2) (i23) {$i_2$};
    \node[inode] at (3.1,-2) (i24) {$i_2$};

    \node[inode] at (1.8,-3) (o1) {$o$};
    \node[inode] at (2.5,-3) (o2) {$o$};
    \node[inode] at (3.5,-3) (o3) {$o$};

    \node[termn] at (0,-4) (t3) {$\varphi_3$};
    \node[termn] at (.7,-4) (t2) {$\varphi_4$};
    \node[terma] at (1.4,-4) (t0) {$\ttrue$};
    \node[terma] at (2.1,-4) (t1) {$\varphi_2$};
    \node[terma] at (2.8,-4) (t4) {$\varphi_1$};
    \node[termn] at (3.5,-4) (t5) {$\varphi_5$};

    \draw[sccr] ($(r2.north west)+(-.5mm,.5mm)$) -| ($(r3.south east)+(.5mm,0)$)
       -- ($(i22.east)+(.5mm,0)$) -- ($(t2.north east)+(.5mm,0)$) |- ($(t3.south west)+(-.5mm,-.5mm)$) -- cycle;
    \draw[scca] ($(r4.north west)+(-.5mm,.5mm)$) -| ($(r1.south east)+(.5mm,0)$)
       -- ($(i24.east)+(.5mm,0)$) -- (t4.east |- o2) -- ($(t4.north east)+(.5mm,0)$) |- ($(t1.south west)+(-.5mm,-.5mm)$) -- ($(t1.north west)+(-.5mm,0)$) -- ($(o1.west)+(-.5mm,0)$) -- ($(i13.west)+(-.5mm,0)$) -- ($(r4.south west)+(-.5mm,0)$) --cycle;
    \draw[scca] ($(r0.north west)+(-.5mm,.5mm)$) -| ($(r0.south east)+(.5mm,0)$) -- ($(o1.west)+(-1mm,0)$) -- ($(t0.north east)+(.5mm,0)$) |- ($(t0.south west)+(-.5mm,-.5mm)$) -- ($(t0.north west)+(-.5mm,0)$) -- ($(i22.east)+(1mm,0)$) -- ($(r0.south west)+(-.5mm,0)$) -- cycle;
    \draw[sccr] ($(r5.north west)+(-.5mm,.5mm)$) -| ($(t5.south east)+(.5mm,-.5mm)$) -| ($(t5.north west)+(-.5mm,0)$) -- ($(o3.west)+(-.5mm,0)$) -- ($(i24.east)+(1mm,0)$) -- ($(r5.south west)+(-.5mm,0)$)  -- cycle;

    \draw[rlink] (r2) -- (i11);
    \draw[rlink] (r3) -- (i12);
    \draw[rlink] (r4) -- (i13);
    \draw[rlink] (r1) -- (i14);
    \draw[rlink] (r5) -- (o3);
    \draw[rlink] (r0) to[out=-90,in=90] (t0.75);

    \draw[low] (i11) to[out=-110,in=90] (t3.120);
    \draw[high] (i11) to[out=-95,in=90] (t2.120);
    \draw[low] (i12) -- (i21);
    \draw[high] (i12) -- (i22);
    \draw[low] (i21) to[out=-70,in=90] (t0.130);
    \draw[high] (i21) to[out=-110,in=90] (t3.90);
    \draw[low] (i22)  to[out=-70,in=90]  (t0.105);
    \draw[high] (i22)  to[out=-110,in=90]  (t2.90);

    \draw[low] (i13) to[bend right=10] (o1);
    \draw[high] (i13) to[bend right=15] (o2);
    \draw[low] (o1) to[out=-170,in=90] (t3.60);
    \draw[high] (o1) to[out=-55,in=90] (t1);

    \draw[low] (i14) -- (i23);
    \draw[high] (i14) -- (i24);

    \draw[low] (i23) -- (o3);
    \draw[high] (i23) -- (o1);
    \draw[low] (i24) -- (o3);
    \draw[high] (i24) -- (o2);

    \draw[low] (o2) to[out=-135,in=90,looseness=1.6] (t2.60);
    \draw[high] (o2) to[out=-55,in=90] (t4);

    \draw[low] (o3) to[out=-135,in=90,looseness=1.1] (t0.50);
    \draw[high] (o3) -- (t5);

  \end{tikzpicture}
  \quad
  \begin{tikzpicture}[initial text={},yscale=0.75,baseline=(current bounding box.center)]

    \node[inode,env] at (0,-1) (i11) {$i_1$};
    \node[inode,env] at (.7,-1) (i12) {$i_1$};
    \node[inode,env,initial above] at (2.1,-1) (i13) {$i_1$};
    \node[inode,env] at (2.8,-1) (i14) {$i_1$};

    \node[inode,env] at (0.4,-2) (i21) {$i_2$};
    \node[inode,env] at (1,-2) (i22) {$i_2$};
    \node[inode,env] at (2.5,-2) (i23) {$i_2$};
    \node[inode,env] at (3.1,-2) (i24) {$i_2$};

    \node[inode,control] at (1.8,-3) (o1) {$o$};
    \node[inode,control] at (2.5,-3) (o2) {$o$};
    \node[inode,control] at (3.5,-3) (o3) {$o$};

    \node[termn] at (0,-4) (t3) {$\varphi_3$};
    \node[termn] at (.7,-4) (t2) {$\varphi_4$};
    \node[terma] at (1.4,-4) (t0) {$\top$};
    \node[terma] at (2.1,-4) (t1) {$\varphi_2$};
    \node[terma] at (2.8,-4) (t4) {$\varphi_1$};
    \node[termn] at (3.5,-4) (t5) {$\varphi_5$};

      \draw[sccr] ($(i11.north west)+(-.5mm,.5mm)$) --
                   ($(i12.north east)+(.5mm,.5mm)$) --
                   ($(i22.north east)+(.5mm,.5mm)$) --
                   ($(i22.south east)+(.5mm,-.5mm)$) --
                   ($(t2.north east)+(.5mm,.5mm)$) |-
                   ($(t3.south west)+(-.5mm,-.5mm)$) -- cycle;

      \draw[scca] ($(i13.north west)+(-.5mm,.5mm)$) --
                   ($(i14.north east)+(.5mm,.5mm)$) --
                   ($(i24.north east)+(.5mm,.5mm)$) --
                   ($(i24.south east)+(.5mm,-.5mm)$) --
                   ($(o2.east)+(.5mm,-.5mm)$) --
                   ($(t4.north east)+(.5mm,.5mm)$) |-
                   ($(t1.south west)+(-.5mm,-.5mm)$) --
                   ($(t1.north west)+(-.5mm,.5mm)$) --
                   ($(o1.west)+(-.5mm,0)$) --
                   ($(i13.south west)+(-.5mm,-.5mm)$) --
                   cycle;
    \draw[scca] ($(t0.north west)+(-.75mm,.75mm)$) -|
    ($(t0.south east)+(.75mm,-.75mm)$) -- ($(t0.south west)+(-.75mm,-.75mm)$) --
    cycle;
    \draw[sccr] ($(t5.north east)+(.5mm,.5mm)$) |-
    ($(t5.south west)+(-.5mm,-.5mm)$) |- ($(o3.north)+(0,-.5mm)$) -| cycle;

    \draw[glink] (i11) to[out=-110,in=90] (t3.120);
    \draw[glink] (i11) to[out=-95,in=90] (t2.120);
    \draw[glink] (i12) -- (i21);
    \draw[glink] (i12) -- (i22);
    \draw[glink] (i21) to[out=-70,in=90] (t0.130);
    \draw[glink] (i21) to[out=-110,in=90] (t3.90);
    \draw[glink] (i22)  to[out=-70,in=90]  (t0.105);
    \draw[glink] (i22)  to[out=-110,in=90]  (t2.90);

    \draw[glink] (i13) to[bend right=10] (o1);
    \draw[glink] (i13) to[bend right=15] (o2);
    \draw[glink] (o1) to[out=-170,in=90] (t3.60);
    \draw[glink] (o1) to[out=-55,in=90] (t1);

    \draw[glink] (i14) -- (i23);
    \draw[glink] (i14) -- (i24);

    \draw[glink] (i23) -- (o3);
    \draw[glink] (i23) -- (o1);
    \draw[glink] (i24) -- (o3);
    \draw[glink] (i24) -- (o2);

    \draw[glink] (o2) to[out=-135,in=90,looseness=1.6] (t2.60);
    \draw[glink] (o2) to[out=-55,in=90] (t4);

    \draw[glink] (o3) to[out=-135,in=90,looseness=1.1] (t0.50);
    \draw[glink] (o3) -- (t5);

    \draw[glink] (t0) edge[loop below] (t0);
    \draw[glink,rounded corners=1mm] (t5.east) -| ++(1mm,5mm) |- ++(-2mm,9mm) -- (o3);
    \draw[glink,rounded corners=1mm] (t4) |- ($(t5.south east)+(2mm,-2mm)$) |-
      ($(i13.north east)+(0mm,+4mm)$) -- (i13);
    \draw[glink,rounded corners=1mm] (t1) |- ($(t5.south east)+(3mm,-3mm)$) |-
      ($(i14.north east)+(0mm,+3mm)$) -- (i14);
    \draw[glink,rounded corners=1mm] (t3) |- ($(t3.south west)+(-2mm,-2mm)$) |-
      ($(i11.north west)+(0mm,+4mm)$) -| (i12);
    \draw[glink,rounded corners=1mm] (t2) |- ($(t3.south west)+(-3mm,-3mm)$) |-
      ($(i11.north west)+(0mm,+3mm)$) -| (i11);
    \end{tikzpicture}
    \quad
    \begin{tikzpicture}[baseline=(current bounding box.center)]
      \node[scca] (SCC1) {$\{\varphi_1,\varphi_2\}$};
      \node[sccr,below=3mm of SCC1.south west](SCC2) {$\{\varphi_3,\varphi_4\}$};
      \node[sccr,below=3mm of SCC1.south east](SCC3) {$\{\varphi_5\}$};
      \node[scca,below=11mm of SCC1.south](SCC4) {$\{\ttrue\}$};
      \draw[->] (SCC1) -- (SCC2);
      \draw[->] (SCC1) -- (SCC3);
      \draw[->] (SCC2) -- (SCC4);
      \draw[->] (SCC3) -- (SCC4);
    \end{tikzpicture}
    \caption[MTDWA and its game interpretation]{The MTDWA from
      Figure~\ref{fig:example}, and its two-player game
      interpretation.  Both also show their decomposition in
      \tikz[baseline=(X.base)]\node[scca](X){accepting}; and
      \tikz[baseline=(X.base)]\node[sccr](X){rejecting}; SCCs.  The
      rightmost picture is the ``condensation graph'' where each SCC
      has been reduced to a single node (labeled by the terminals it
      contains for easy identification).\label{fig:game-interpretation}}
\end{figure}
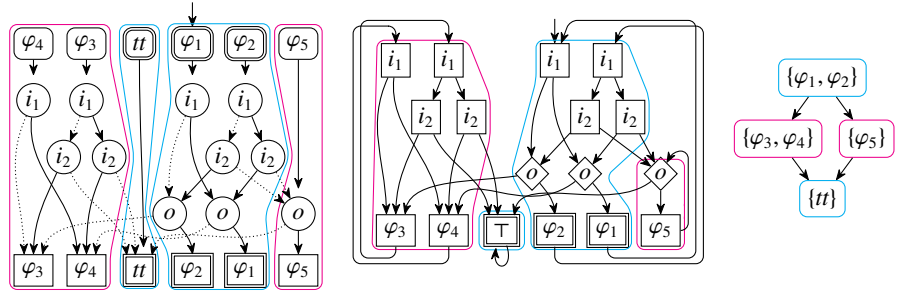

Figure~\ref{fig:game-interpretation} shows how an MTDWA can be
interpreted as a two-player game with weak objective.  Here we are
doing synthesis with Mealy semantics, so the input variables are
ordered above the output variables.

In the game interpretation, each internal node is assigned either to
the input player (rectangular nodes), or to the output player
(diamond nodes) based on the nature of the variable labeling that
node.  The terminal nodes can be assigned to either player since they
do not allow any choice: from a terminal $\boxed{\alpha}$, the only
possible move is to jump to the node corresponding to the root of
$\Delta(\alpha)$.

The game is won by the output player iff it has a strategy to force
the game to reach the accepting terminal ($\ttrue$, $\varphi_1$, or
$\varphi_2$) infinitely often.  In terms of SCCs, the game is won by
the output player iff it has a strategy to force the game to get stuck
into an accepting SCC ($\{\ttrue\}$, or $\{\varphi_1,\varphi_2\}$).

Before we solve the game, let us first address the fact that the
condensation graph shown here does not contain an edge between
$\varphi_2$ and $\ttrue$, despite the fact that the explicit automaton
from Figure~\ref{fig:example} has one.  If we are in state $\varphi_2$
and the input player plays $i_1i_2$ or $\bar{i_1}i_2$, then from the
point of view of the output player, the game can be played exactly as
if we were in state $\varphi_5$.  The MTDWA representation captures
this, and the paths between $\varphi_2$ and $\top$ will temporarily
enter the SCC of $\varphi_5$. This is inconsequential in practice
since it will not affect the topological ordering of the SCCs.

\begin{figure}[tbp]
  \centering
  \begin{subfigure}{0.32\textwidth}\centering
  \begin{tikzpicture}[initial text={},yscale=0.75,baseline=(current bounding box.center)]
    \node[root, accepting, initial above] at (2.1,0) (r4) {$\varphi_1$};

    \node[inode] at (2.1,-1) (i13) {$i_1$};


    \node[inode] at (1.8,-3) (o1) {$o$};
    \node[inode] at (2.5,-3) (o2) {$o$};

    \node[termn] at (0,-4) (t3) {$\varphi_3$};
    \node[termn] at (.7,-4) (t2) {$\varphi_4$};
    \node[terma] at (2.1,-4) (t1) {$\varphi_2$};
    \node[terma] at (2.8,-4) (t4) {$\varphi_1$};


    \draw[rlink] (r4) -- (i13);


    \draw[low] (i13) to[bend right=10] (o1);
    \draw[high] (i13) to[bend right=15] (o2);
    \draw[low] (o1) to[out=-170,in=90] (t3.60);
    \draw[high] (o1) to[out=-55,in=90] (t1);



    \draw[low] (o2) to[out=-135,in=90,looseness=1.6] (t2.60);
    \draw[high] (o2) to[out=-55,in=90] (t4);


  \end{tikzpicture}
  \caption{$\Delta(\varphi_1)$ computed}
\end{subfigure}
  \begin{subfigure}{0.32\textwidth}\centering
  \begin{tikzpicture}[initial text={},yscale=0.75,baseline=(current bounding box.center)]
    \node[root] at (0,0) (r2) {$\varphi_4$};
    \node[root, accepting, initial above] at (2.1,0) (r4) {$\varphi_1$};

    \node[inode] at (0,-1) (i11) {$i_1$};
    \node[inode] at (2.1,-1) (i13) {$i_1$};


    \node[inode] at (1.8,-3) (o1) {$o$};
    \node[inode] at (2.5,-3) (o2) {$o$};

    \node[termn] at (0,-4) (t3) {$\varphi_3$};
    \node[termn] at (.7,-4) (t2) {$\varphi_4$};
    \node[terma] at (2.1,-4) (t1) {$\varphi_2$};
    \node[terma] at (2.8,-4) (t4) {$\varphi_1$};


     \draw[rlink] (r2) -- (i11);
    \draw[rlink] (r4) -- (i13);

    \draw[low] (i11) to[out=-110,in=90] (t3.120);
    \draw[high] (i11) to[out=-95,in=90] (t2.120);

    \draw[low] (i13) to[bend right=10] (o1);
    \draw[high] (i13) to[bend right=15] (o2);
    \draw[low] (o1) to[out=-170,in=90] (t3.60);
    \draw[high] (o1) to[out=-55,in=90] (t1);



    \draw[low] (o2) to[out=-135,in=90,looseness=1.6] (t2.60);
    \draw[high] (o2) to[out=-55,in=90] (t4);


  \end{tikzpicture}
  \caption{$\Delta(\varphi_4)$ computed}
\end{subfigure}
  \begin{subfigure}{0.32\textwidth}\centering
  \begin{tikzpicture}[initial text={},yscale=0.75,baseline=(current bounding box.center)]
    \node[root] at (0,0) (r2) {$\varphi_4$};
    \node[root] at (.7,0) (r3) {$\varphi_3$};
    \node[root, accepting, initial above] at (2.1,0) (r4) {$\varphi_1$};

    \node[inode] at (0,-1) (i11) {$i_1$};
    \node[inode] at (.7,-1) (i12) {$i_1$};
    \node[inode] at (2.1,-1) (i13) {$i_1$};

    \node[inode] at (0.4,-2) (i21) {$i_2$};
    \node[inode] at (1,-2) (i22) {$i_2$};

    \node[inode] at (1.8,-3) (o1) {$o$};
    \node[inode] at (2.5,-3) (o2) {$o$};

    \node[termn] at (0,-4) (t3) {$\varphi_3$};
    \node[termn] at (.7,-4) (t2) {$\varphi_4$};
    \node[terma] at (1.4,-4) (t0) {$\ttrue$};
    \node[terma] at (2.1,-4) (t1) {$\varphi_2$};
    \node[terma] at (2.8,-4) (t4) {$\varphi_1$};


     \draw[rlink] (r2) -- (i11);
    \draw[rlink] (r3) -- (i12);
    \draw[rlink] (r4) -- (i13);

    \draw[low] (i11) to[out=-110,in=90] (t3.120);
    \draw[high] (i11) to[out=-95,in=90] (t2.120);
    \draw[low] (i12) -- (i21);
    \draw[high] (i12) -- (i22);
    \draw[low] (i21) to[out=-70,in=90] (t0.130);
    \draw[high] (i21) to[out=-110,in=90] (t3.90);
    \draw[low] (i22)  to[out=-70,in=90]  (t0.105);
    \draw[high] (i22)  to[out=-110,in=90]  (t2.90);

    \draw[low] (i13) to[bend right=10] (o1);
    \draw[high] (i13) to[bend right=15] (o2);
    \draw[low] (o1) to[out=-170,in=90] (t3.60);
    \draw[high] (o1) to[out=-55,in=90] (t1);



    \draw[low] (o2) to[out=-135,in=90,looseness=1.6] (t2.60);
    \draw[high] (o2) to[out=-55,in=90] (t4);


  \end{tikzpicture}
  \caption{$\Delta(\varphi_3)$ computed}
\end{subfigure}
  \begin{subfigure}{0.32\textwidth}\centering
  \begin{tikzpicture}[initial text={},yscale=0.75,baseline=(current bounding box.center)]
    \node[root] at (0,0) (r2) {$\varphi_4$};
    \node[root] at (.7,0) (r3) {$\varphi_3$};
    \node[root, accepting] at (1.4,0) (r0) {$\ttrue$};
    \node[root, accepting, initial above] at (2.1,0) (r4) {$\varphi_1$};

    \node[inode] at (0,-1) (i11) {$i_1$};
    \node[inode] at (.7,-1) (i12) {$i_1$};
    \node[inode] at (2.1,-1) (i13) {$i_1$};

    \node[inode] at (0.4,-2) (i21) {$i_2$};
    \node[inode] at (1,-2) (i22) {$i_2$};

    \node[inode] at (1.8,-3) (o1) {$o$};
    \node[inode] at (2.5,-3) (o2) {$o$};

    \node[termn] at (0,-4) (t3) {$\varphi_3$};
    \node[termn] at (.7,-4) (t2) {$\varphi_4$};
    \node[terma] at (1.4,-4) (t0) {$\ttrue$};
    \node[terma] at (2.1,-4) (t1) {$\varphi_2$};
    \node[terma] at (2.8,-4) (t4) {$\varphi_1$};


     \draw[rlink] (r2) -- (i11);
    \draw[rlink] (r3) -- (i12);
    \draw[rlink] (r4) -- (i13);
    \draw[rlink] (r0) to[out=-90,in=90] (t0.75);

    \draw[low] (i11) to[out=-110,in=90] (t3.120);
    \draw[high] (i11) to[out=-95,in=90] (t2.120);
    \draw[low] (i12) -- (i21);
    \draw[high] (i12) -- (i22);
    \draw[low] (i21) to[out=-70,in=90] (t0.130);
    \draw[high] (i21) to[out=-110,in=90] (t3.90);
    \draw[low] (i22)  to[out=-70,in=90]  (t0.105);
    \draw[high] (i22)  to[out=-110,in=90]  (t2.90);

    \draw[low] (i13) to[bend right=10] (o1);
    \draw[high] (i13) to[bend right=15] (o2);
    \draw[low] (o1) to[out=-170,in=90] (t3.60);
    \draw[high] (o1) to[out=-55,in=90] (t1);



    \draw[low] (o2) to[out=-135,in=90,looseness=1.6] (t2.60);
    \draw[high] (o2) to[out=-55,in=90] (t4);


  \end{tikzpicture}
  \caption{$\Delta(\ttrue)$ computed}
\end{subfigure}
  \begin{subfigure}{0.32\textwidth}\centering
  \begin{tikzpicture}[initial text={},yscale=0.75,baseline=(current bounding box.center)]
    \node[root] at (0,0) (r2) {$\varphi_4$};
    \node[root] at (.7,0) (r3) {$\varphi_3$};
    \node[root, accepting,win] at (1.4,0) (r0) {$\ttrue$};
    \node[root, accepting, initial above] at (2.1,0) (r4) {$\varphi_1$};

    \node[inode] at (0,-1) (i11) {$i_1$};
    \node[inode] at (.7,-1) (i12) {$i_1$};
    \node[inode] at (2.1,-1) (i13) {$i_1$};

    \node[inode] at (0.4,-2) (i21) {$i_2$};
    \node[inode] at (1,-2) (i22) {$i_2$};

    \node[inode] at (1.8,-3) (o1) {$o$};
    \node[inode] at (2.5,-3) (o2) {$o$};

    \node[termn] at (0,-4) (t3) {$\varphi_3$};
    \node[termn] at (.7,-4) (t2) {$\varphi_4$};
    \node[terma,win] at (1.4,-4) (t0) {$\ttrue$};
    \node[terma] at (2.1,-4) (t1) {$\varphi_2$};
    \node[terma] at (2.8,-4) (t4) {$\varphi_1$};

    \draw[scca] ($(r0.north west)+(-.5mm,.5mm)$) -| ($(r0.south east)+(.5mm,0)$) -- ($(o1.west)+(-1mm,0)$) -- ($(t0.north east)+(.5mm,0)$) |- ($(t0.south west)+(-.5mm,-.5mm)$) -- ($(t0.north west)+(-.5mm,0)$) -- ($(i22.east)+(1mm,0)$) -- ($(r0.south west)+(-.5mm,0)$) -- cycle;

     \draw[rlink] (r2) -- (i11);
    \draw[rlink] (r3) -- (i12);
    \draw[rlink] (r4) -- (i13);
    \draw[rlink] (r0) to[out=-90,in=90] (t0.75);

    \draw[low] (i11) to[out=-110,in=90] (t3.120);
    \draw[high] (i11) to[out=-95,in=90] (t2.120);
    \draw[low] (i12) -- (i21);
    \draw[high] (i12) -- (i22);
    \draw[low] (i21) to[out=-70,in=90] (t0.130);
    \draw[high] (i21) to[out=-110,in=90] (t3.90);
    \draw[low] (i22)  to[out=-70,in=90]  (t0.105);
    \draw[high] (i22)  to[out=-110,in=90]  (t2.90);

    \draw[low] (i13) to[bend right=10] (o1);
    \draw[high] (i13) to[bend right=15] (o2);
    \draw[low] (o1) to[out=-170,in=90] (t3.60);
    \draw[high] (o1) to[out=-55,in=90] (t1);



    \draw[low] (o2) to[out=-135,in=90,looseness=1.6] (t2.60);
    \draw[high] (o2) to[out=-55,in=90] (t4);


  \end{tikzpicture}
  \caption{SCC $\{\ttrue\}$ backtracked}
\end{subfigure}
  \begin{subfigure}{0.32\textwidth}\centering
  \begin{tikzpicture}[initial text={},yscale=0.75,baseline=(current bounding box.center)]
    \node[root,lose] at (0,0) (r2) {$\varphi_4$};
    \node[root,lose] at (.7,0) (r3) {$\varphi_3$};
    \node[root, accepting,win] at (1.4,0) (r0) {$\ttrue$};
    \node[root, accepting, initial above] at (2.1,0) (r4) {$\varphi_1$};

    \node[inode,lose] at (0,-1) (i11) {$i_1$};
    \node[inode,lose] at (.7,-1) (i12) {$i_1$};
    \node[inode] at (2.1,-1) (i13) {$i_1$};

    \node[inode,lose] at (0.4,-2) (i21) {$i_2$};
    \node[inode,lose] at (1,-2) (i22) {$i_2$};

    \node[inode] at (1.8,-3) (o1) {$o$};
    \node[inode] at (2.5,-3) (o2) {$o$};

    \node[termn,lose] at (0,-4) (t3) {$\varphi_3$};
    \node[termn,lose] at (.7,-4) (t2) {$\varphi_4$};
    \node[terma,win] at (1.4,-4) (t0) {$\ttrue$};
    \node[terma] at (2.1,-4) (t1) {$\varphi_2$};
    \node[terma] at (2.8,-4) (t4) {$\varphi_1$};

    \draw[sccr] ($(r2.north west)+(-.5mm,.5mm)$) -| ($(r3.south east)+(.5mm,0)$) -- ($(i22.east)+(.5mm,0)$) -- ($(t2.north east)+(.5mm,0)$) |- ($(t3.south west)+(-.5mm,-.5mm)$) -- cycle;
    \draw[scca] ($(r0.north west)+(-.5mm,.5mm)$) -| ($(r0.south east)+(.5mm,0)$) -- ($(o1.west)+(-1mm,0)$) -- ($(t0.north east)+(.5mm,0)$) |- ($(t0.south west)+(-.5mm,-.5mm)$) -- ($(t0.north west)+(-.5mm,0)$) -- ($(i22.east)+(1mm,0)$) -- ($(r0.south west)+(-.5mm,0)$) -- cycle;

     \draw[rlink] (r2) -- (i11);
    \draw[rlink] (r3) -- (i12);
    \draw[rlink] (r4) -- (i13);
    \draw[rlink] (r0) to[out=-90,in=90] (t0.75);

    \draw[low] (i11) to[out=-110,in=90] (t3.120);
    \draw[high] (i11) to[out=-95,in=90] (t2.120);
    \draw[low] (i12) -- (i21);
    \draw[high] (i12) -- (i22);
    \draw[low] (i21) to[out=-70,in=90] (t0.130);
    \draw[high] (i21) to[out=-110,in=90] (t3.90);
    \draw[low] (i22)  to[out=-70,in=90]  (t0.105);
    \draw[high] (i22)  to[out=-110,in=90]  (t2.90);

    \draw[low] (i13) to[bend right=10] (o1);
    \draw[high] (i13) to[bend right=15] (o2);
    \draw[low] (o1) to[out=-170,in=90] (t3.60);
    \draw[high] (o1) to[out=-55,in=90] (t1);



    \draw[low] (o2) to[out=-135,in=90,looseness=1.6] (t2.60);
    \draw[high] (o2) to[out=-55,in=90] (t4);


  \end{tikzpicture}
  \caption{SCC $\{\varphi_3,\varphi_4\}$ backtracked}
\end{subfigure}

\begin{subfigure}{0.45\textwidth}\centering
  \begin{tikzpicture}[initial text={},yscale=0.75,baseline=(current bounding box.center)]
    \node[root,lose] at (0,0) (r2) {$\varphi_4$};
    \node[root,lose] at (.7,0) (r3) {$\varphi_3$};
    \node[root, accepting,win] at (1.4,0) (r0) {$\ttrue$};
    \node[root, accepting, initial above] at (2.1,0) (r4) {$\varphi_1$};
    \node[root, accepting] at (2.8,0) (r1) {$\varphi_2$};

    \node[inode,lose] at (0,-1) (i11) {$i_1$};
    \node[inode,lose] at (.7,-1) (i12) {$i_1$};
    \node[inode] at (2.1,-1) (i13) {$i_1$};
    \node[inode] at (2.8,-1) (i14) {$i_1$};

    \node[inode,lose] at (0.4,-2) (i21) {$i_2$};
    \node[inode,lose] at (1,-2) (i22) {$i_2$};
    \node[inode] at (2.5,-2) (i23) {$i_2$};
    \node[inode] at (3.1,-2) (i24) {$i_2$};

    \node[inode] at (1.8,-3) (o1) {$o$};
    \node[inode] at (2.5,-3) (o2) {$o$};
    \node[inode,win] at (3.5,-3) (o3) {$o$};

    \node[termn,lose] at (0,-4) (t3) {$\varphi_3$};
    \node[termn,lose] at (.7,-4) (t2) {$\varphi_4$};
    \node[terma,win] at (1.4,-4) (t0) {$\ttrue$};
    \node[terma] at (2.1,-4) (t1) {$\varphi_2$};
    \node[terma] at (2.8,-4) (t4) {$\varphi_1$};
    \node[termn] at (3.5,-4) (t5) {$\varphi_5$};

    \draw[sccr] ($(r2.north west)+(-.5mm,.5mm)$) -| ($(r3.south east)+(.5mm,0)$) -- ($(i22.east)+(.5mm,0)$) -- ($(t2.north east)+(.5mm,0)$) |- ($(t3.south west)+(-.5mm,-.5mm)$) -- cycle;
    \draw[scca] ($(r0.north west)+(-.5mm,.5mm)$) -| ($(r0.south east)+(.5mm,0)$) -- ($(o1.west)+(-1mm,0)$) -- ($(t0.north east)+(.5mm,0)$) |- ($(t0.south west)+(-.5mm,-.5mm)$) -- ($(t0.north west)+(-.5mm,0)$) -- ($(i22.east)+(1mm,0)$) -- ($(r0.south west)+(-.5mm,0)$) -- cycle;

     \draw[rlink] (r2) -- (i11);
    \draw[rlink] (r3) -- (i12);
    \draw[rlink] (r4) -- (i13);
    \draw[rlink] (r1) -- (i14);
    \draw[rlink] (r0) to[out=-90,in=90] (t0.75);

    \draw[low] (i11) to[out=-110,in=90] (t3.120);
    \draw[high] (i11) to[out=-95,in=90] (t2.120);
    \draw[low] (i12) -- (i21);
    \draw[high] (i12) -- (i22);
    \draw[low] (i21) to[out=-70,in=90] (t0.130);
    \draw[high] (i21) to[out=-110,in=90] (t3.90);
    \draw[low] (i22)  to[out=-70,in=90]  (t0.105);
    \draw[high] (i22)  to[out=-110,in=90]  (t2.90);

    \draw[low] (i13) to[bend right=10] (o1);
    \draw[high] (i13) to[bend right=15] (o2);
    \draw[low] (o1) to[out=-170,in=90] (t3.60);
    \draw[high] (o1) to[out=-55,in=90] (t1);

    \draw[low] (i14) -- (i23);
    \draw[high] (i14) -- (i24);

    \draw[low] (i23) -- (o3);
    \draw[high] (i23) -- (o1);
    \draw[low] (i24) -- (o3);
    \draw[high] (i24) -- (o2);

    \draw[low] (o2) to[out=-135,in=90,looseness=1.6] (t2.60);
    \draw[high] (o2) to[out=-55,in=90] (t4);

    \draw[low,winedge] (o3) to[out=-135,in=90,looseness=1.1] (t0.50);
    \draw[high] (o3) -- (t5);

  \end{tikzpicture}
  \caption{$\Delta(\varphi_2)$ computed}
\end{subfigure}
\begin{subfigure}{0.45\textwidth}\centering
  \begin{tikzpicture}[initial text={},yscale=0.75,baseline=(current bounding box.center)]
    \node[root,lose] at (0,0) (r2) {$\varphi_4$};
    \node[root,lose] at (.7,0) (r3) {$\varphi_3$};
    \node[root, accepting,win] at (1.4,0) (r0) {$\ttrue$};
    \node[root, accepting, initial above] at (2.1,0) (r4) {$\varphi_1$};
    \node[root, accepting] at (2.8,0) (r1) {$\varphi_2$};
    \node[root,win] at (3.5,0) (r5) {$\varphi_5$};

    \node[inode,lose] at (0,-1) (i11) {$i_1$};
    \node[inode,lose] at (.7,-1) (i12) {$i_1$};
    \node[inode] at (2.1,-1) (i13) {$i_1$};
    \node[inode] at (2.8,-1) (i14) {$i_1$};

    \node[inode,lose] at (0.4,-2) (i21) {$i_2$};
    \node[inode,lose] at (1,-2) (i22) {$i_2$};
    \node[inode] at (2.5,-2) (i23) {$i_2$};
    \node[inode] at (3.1,-2) (i24) {$i_2$};

    \node[inode] at (1.8,-3) (o1) {$o$};
    \node[inode] at (2.5,-3) (o2) {$o$};
    \node[inode,win] at (3.5,-3) (o3) {$o$};

    \node[termn,lose] at (0,-4) (t3) {$\varphi_3$};
    \node[termn,lose] at (.7,-4) (t2) {$\varphi_4$};
    \node[terma,win] at (1.4,-4) (t0) {$\ttrue$};
    \node[terma] at (2.1,-4) (t1) {$\varphi_2$};
    \node[terma] at (2.8,-4) (t4) {$\varphi_1$};
    \node[termn,win] at (3.5,-4) (t5) {$\varphi_5$};

    \draw[sccr] ($(r2.north west)+(-.5mm,.5mm)$) -| ($(r3.south east)+(.5mm,0)$) -- ($(i22.east)+(.5mm,0)$) -- ($(t2.north east)+(.5mm,0)$) |- ($(t3.south west)+(-.5mm,-.5mm)$) -- cycle;
    \draw[scca] ($(r0.north west)+(-.5mm,.5mm)$) -| ($(r0.south east)+(.5mm,0)$) -- ($(o1.west)+(-1mm,0)$) -- ($(t0.north east)+(.5mm,0)$) |- ($(t0.south west)+(-.5mm,-.5mm)$) -- ($(t0.north west)+(-.5mm,0)$) -- ($(i22.east)+(1mm,0)$) -- ($(r0.south west)+(-.5mm,0)$) -- cycle;

     \draw[rlink] (r2) -- (i11);
    \draw[rlink] (r3) -- (i12);
    \draw[rlink] (r4) -- (i13);
    \draw[rlink] (r1) -- (i14);
    \draw[rlink] (r5) -- (o3);
    \draw[rlink] (r0) to[out=-90,in=90] (t0.75);

    \draw[low] (i11) to[out=-110,in=90] (t3.120);
    \draw[high] (i11) to[out=-95,in=90] (t2.120);
    \draw[low] (i12) -- (i21);
    \draw[high] (i12) -- (i22);
    \draw[low] (i21) to[out=-70,in=90] (t0.130);
    \draw[high] (i21) to[out=-110,in=90] (t3.90);
    \draw[low] (i22)  to[out=-70,in=90]  (t0.105);
    \draw[high] (i22)  to[out=-110,in=90]  (t2.90);

    \draw[low] (i13) to[bend right=10] (o1);
    \draw[high] (i13) to[bend right=15] (o2);
    \draw[low] (o1) to[out=-170,in=90] (t3.60);
    \draw[high] (o1) to[out=-55,in=90] (t1);

    \draw[low] (i14) -- (i23);
    \draw[high] (i14) -- (i24);

    \draw[low] (i23) -- (o3);
    \draw[high] (i23) -- (o1);
    \draw[low] (i24) -- (o3);
    \draw[high] (i24) -- (o2);

    \draw[low] (o2) to[out=-135,in=90,looseness=1.6] (t2.60);
    \draw[high] (o2) to[out=-55,in=90] (t4);

    \draw[low,winedge] (o3) to[out=-135,in=90,looseness=1.1] (t0.50);
    \draw[high] (o3) -- (t5);

  \end{tikzpicture}
  \caption{$\Delta(\varphi_5)$ computed}
\end{subfigure}

  \begin{subfigure}{0.45\textwidth}\centering
  \begin{tikzpicture}[initial text={},yscale=0.75,baseline=(current bounding box.center)]
    \node[root,lose] at (0,0) (r2) {$\varphi_4$};
    \node[root,lose] at (.7,0) (r3) {$\varphi_3$};
    \node[root, accepting,win] at (1.4,0) (r0) {$\ttrue$};
    \node[root, accepting, initial above] at (2.1,0) (r4) {$\varphi_1$};
    \node[root, accepting] at (2.8,0) (r1) {$\varphi_2$};
    \node[root,win] at (3.5,0) (r5) {$\varphi_5$};

    \node[inode,lose] at (0,-1) (i11) {$i_1$};
    \node[inode,lose] at (.7,-1) (i12) {$i_1$};
    \node[inode] at (2.1,-1) (i13) {$i_1$};
    \node[inode] at (2.8,-1) (i14) {$i_1$};

    \node[inode,lose] at (0.4,-2) (i21) {$i_2$};
    \node[inode,lose] at (1,-2) (i22) {$i_2$};
    \node[inode] at (2.5,-2) (i23) {$i_2$};
    \node[inode] at (3.1,-2) (i24) {$i_2$};

    \node[inode] at (1.8,-3) (o1) {$o$};
    \node[inode] at (2.5,-3) (o2) {$o$};
    \node[inode,win] at (3.5,-3) (o3) {$o$};

    \node[termn,lose] at (0,-4) (t3) {$\varphi_3$};
    \node[termn,lose] at (.7,-4) (t2) {$\varphi_4$};
    \node[terma,win] at (1.4,-4) (t0) {$\ttrue$};
    \node[terma] at (2.1,-4) (t1) {$\varphi_2$};
    \node[terma] at (2.8,-4) (t4) {$\varphi_1$};
    \node[termn,win] at (3.5,-4) (t5) {$\varphi_5$};

    \draw[sccr] ($(r2.north west)+(-.5mm,.5mm)$) -| ($(r3.south east)+(.5mm,0)$) -- ($(i22.east)+(.5mm,0)$) -- ($(t2.north east)+(.5mm,0)$) |- ($(t3.south west)+(-.5mm,-.5mm)$) -- cycle;
    \draw[scca] ($(r0.north west)+(-.5mm,.5mm)$) -| ($(r0.south east)+(.5mm,0)$) -- ($(o1.west)+(-1mm,0)$) -- ($(t0.north east)+(.5mm,0)$) |- ($(t0.south west)+(-.5mm,-.5mm)$) -- ($(t0.north west)+(-.5mm,0)$) -- ($(i22.east)+(1mm,0)$) -- ($(r0.south west)+(-.5mm,0)$) -- cycle;
    \draw[sccr] ($(r5.north west)+(-.5mm,.5mm)$) -| ($(t5.south east)+(.5mm,-.5mm)$) -| ($(t5.north west)+(-.5mm,0)$) -- ($(o3.west)+(-.5mm,0)$) -- ($(i24.east)+(1mm,0)$) -- ($(r5.south west)+(-.5mm,0)$)  -- cycle;

     \draw[rlink] (r2) -- (i11);
    \draw[rlink] (r3) -- (i12);
    \draw[rlink] (r4) -- (i13);
    \draw[rlink] (r1) -- (i14);
    \draw[rlink] (r5) -- (o3);
    \draw[rlink] (r0) to[out=-90,in=90] (t0.75);

    \draw[low] (i11) to[out=-110,in=90] (t3.120);
    \draw[high] (i11) to[out=-95,in=90] (t2.120);
    \draw[low] (i12) -- (i21);
    \draw[high] (i12) -- (i22);
    \draw[low] (i21) to[out=-70,in=90] (t0.130);
    \draw[high] (i21) to[out=-110,in=90] (t3.90);
    \draw[low] (i22)  to[out=-70,in=90]  (t0.105);
    \draw[high] (i22)  to[out=-110,in=90]  (t2.90);

    \draw[low] (i13) to[bend right=10] (o1);
    \draw[high] (i13) to[bend right=15] (o2);
    \draw[low] (o1) to[out=-170,in=90] (t3.60);
    \draw[high] (o1) to[out=-55,in=90] (t1);

    \draw[low] (i14) -- (i23);
    \draw[high] (i14) -- (i24);

    \draw[low] (i23) -- (o3);
    \draw[high] (i23) -- (o1);
    \draw[low] (i24) -- (o3);
    \draw[high] (i24) -- (o2);

    \draw[low] (o2) to[out=-135,in=90,looseness=1.6] (t2.60);
    \draw[high] (o2) to[out=-55,in=90] (t4);

    \draw[low,winedge] (o3) to[out=-135,in=90,looseness=1.1] (t0.50);
    \draw[high] (o3) -- (t5);

  \end{tikzpicture}
  \caption{SCC $\{\varphi_5\}$ backtracked}
\end{subfigure}
  \begin{subfigure}{0.45\textwidth}\centering
  \begin{tikzpicture}[initial text={},yscale=0.75,baseline=(current bounding box.center)]
    \node[root,lose] at (0,0) (r2) {$\varphi_4$};
    \node[root,lose] at (.7,0) (r3) {$\varphi_3$};
    \node[root, accepting,win] at (1.4,0) (r0) {$\ttrue$};
    \node[root, accepting,win, initial above] at (2.1,0) (r4) {$\varphi_1$};
    \node[root, accepting,win] at (2.8,0) (r1) {$\varphi_2$};
    \node[root,win] at (3.5,0) (r5) {$\varphi_5$};

    \node[inode,lose] at (0,-1) (i11) {$i_1$};
    \node[inode,lose] at (.7,-1) (i12) {$i_1$};
    \node[inode,win] at (2.1,-1) (i13) {$i_1$};
    \node[inode,win] at (2.8,-1) (i14) {$i_1$};

    \node[inode,lose] at (0.4,-2) (i21) {$i_2$};
    \node[inode,lose] at (1,-2) (i22) {$i_2$};
    \node[inode,win] at (2.5,-2) (i23) {$i_2$};
    \node[inode,win] at (3.1,-2) (i24) {$i_2$};

    \node[inode,win] at (1.8,-3) (o1) {$o$};
    \node[inode,win] at (2.5,-3) (o2) {$o$};
    \node[inode,win] at (3.5,-3) (o3) {$o$};

    \node[termn,lose] at (0,-4) (t3) {$\varphi_3$};
    \node[termn,lose] at (.7,-4) (t2) {$\varphi_4$};
    \node[terma,win] at (1.4,-4) (t0) {$\ttrue$};
    \node[terma,win] at (2.1,-4) (t1) {$\varphi_2$};
    \node[terma,win] at (2.8,-4) (t4) {$\varphi_1$};
    \node[termn,win] at (3.5,-4) (t5) {$\varphi_5$};

    \draw[sccr] ($(r2.north west)+(-.5mm,.5mm)$) -| ($(r3.south east)+(.5mm,0)$) -- ($(i22.east)+(.5mm,0)$) -- ($(t2.north east)+(.5mm,0)$) |- ($(t3.south west)+(-.5mm,-.5mm)$) -- cycle;
    \draw[scca] ($(r4.north west)+(-.5mm,.5mm)$) -| ($(r1.south east)+(.5mm,0)$) -- ($(i24.east)+(.5mm,0)$) -- (t4.east |- o2) -- ($(t4.north east)+(.5mm,0)$) |- ($(t1.south west)+(-.5mm,-.5mm)$) -- ($(t1.north west)+(-.5mm,0)$) -- ($(o1.west)+(-.5mm,0)$) -- ($(i13.west)+(-.5mm,0)$) -- ($(r4.south west)+(-.5mm,0)$) --cycle;
    \draw[scca] ($(r0.north west)+(-.5mm,.5mm)$) -| ($(r0.south east)+(.5mm,0)$) -- ($(o1.west)+(-1mm,0)$) -- ($(t0.north east)+(.5mm,0)$) |- ($(t0.south west)+(-.5mm,-.5mm)$) -- ($(t0.north west)+(-.5mm,0)$) -- ($(i22.east)+(1mm,0)$) -- ($(r0.south west)+(-.5mm,0)$) -- cycle;
    \draw[sccr] ($(r5.north west)+(-.5mm,.5mm)$) -| ($(t5.south east)+(.5mm,-.5mm)$) -| ($(t5.north west)+(-.5mm,0)$) -- ($(o3.west)+(-.5mm,0)$) -- ($(i24.east)+(1mm,0)$) -- ($(r5.south west)+(-.5mm,0)$)  -- cycle;

     \draw[rlink] (r2) -- (i11);
    \draw[rlink] (r3) -- (i12);
    \draw[rlink] (r4) -- (i13);
    \draw[rlink] (r1) -- (i14);
    \draw[rlink] (r5) -- (o3);
    \draw[rlink] (r0) to[out=-90,in=90] (t0.75);

    \draw[low] (i11) to[out=-110,in=90] (t3.120);
    \draw[high] (i11) to[out=-95,in=90] (t2.120);
    \draw[low] (i12) -- (i21);
    \draw[high] (i12) -- (i22);
    \draw[low] (i21) to[out=-70,in=90] (t0.130);
    \draw[high] (i21) to[out=-110,in=90] (t3.90);
    \draw[low] (i22)  to[out=-70,in=90]  (t0.105);
    \draw[high] (i22)  to[out=-110,in=90]  (t2.90);

    \draw[low] (i13) to[bend right=10] (o1);
    \draw[high] (i13) to[bend right=15] (o2);
    \draw[low] (o1) to[out=-170,in=90] (t3.60);
    \draw[high,winedge] (o1) to[out=-55,in=90] (t1);

    \draw[low] (i14) -- (i23);
    \draw[high] (i14) -- (i24);

    \draw[low] (i23) -- (o3);
    \draw[high] (i23) -- (o1);
    \draw[low] (i24) -- (o3);
    \draw[high] (i24) -- (o2);

    \draw[low] (o2) to[out=-135,in=90,looseness=1.6] (t2.60);
    \draw[high,winedge] (o2) to[out=-55,in=90] (t4);

    \draw[low,winedge] (o3) to[out=-135,in=90,looseness=1.1] (t0.50);
    \draw[high] (o3) -- (t5);

  \end{tikzpicture}
  \caption{SCC $\{\varphi_1,\varphi_2\}$ backtracked}
\end{subfigure}
\caption{Solving the example from Figure~\ref{fig:example} as a game, while it is being constructed.\label{fig:otf}}
\end{figure}
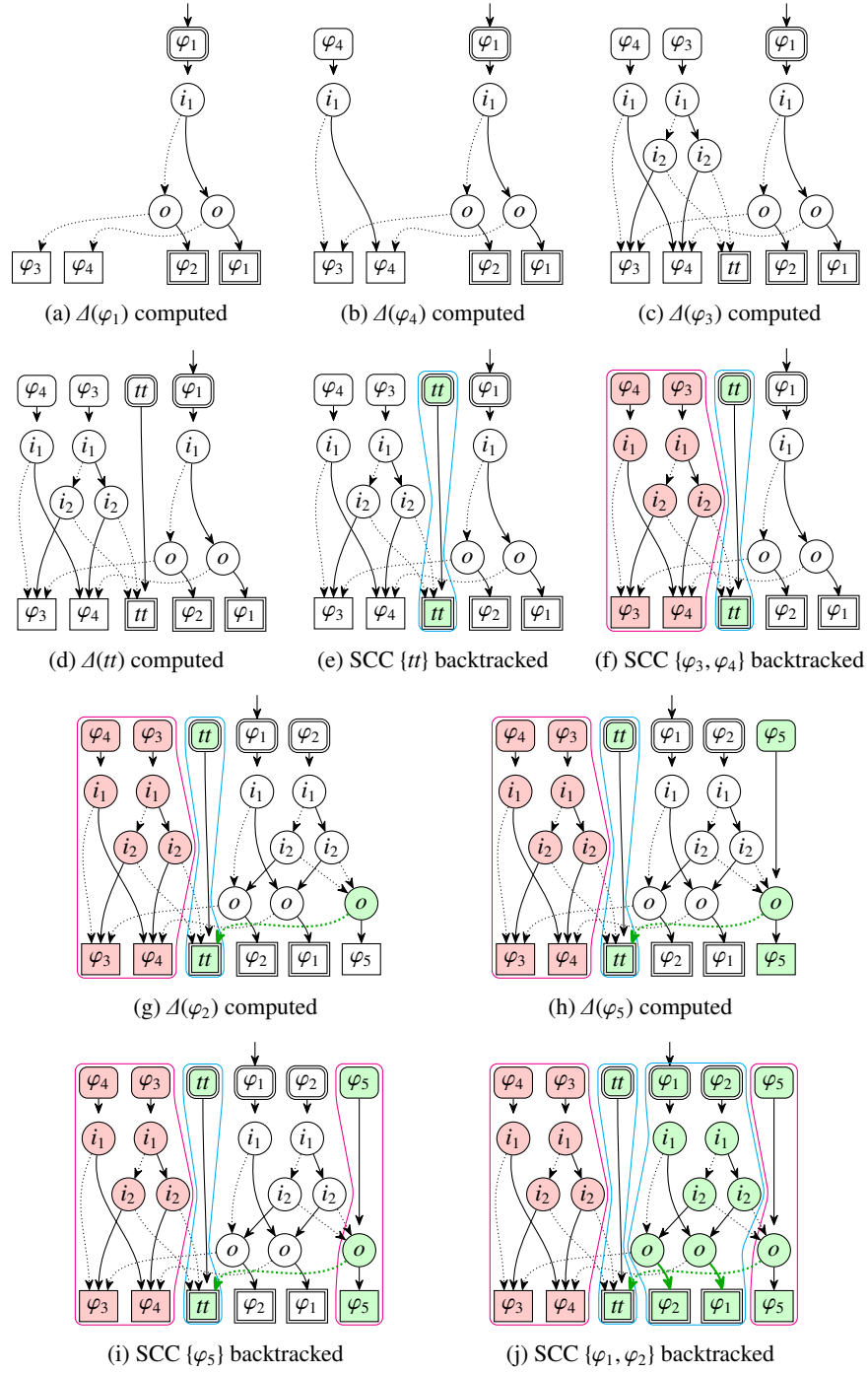

Figure~\ref{fig:otf} shows the step-by-step construction of our
example MTDWA, solving it as a game along the way.  The automaton is
explored in DFS order because the goal is to find strongly-connected
components.

In the first steps (a)--(d), the DFS explores states $\varphi_1$,
$\varphi_4$, $\varphi_3$, and $\ttrue$ in that order.  For each of
those, $\Delta(\varphi_i)$ is computed (remember that
$\Delta(\varphi_i)$ is $\tr(\varphi_i)$ from Section~\ref{sec:tr},
followed by propositional equivalence from Section~\ref{sec:propeq}).
During this exploration, the SCC algorithm updates its data structures
(not shown nor discussed here) for the edges that are discovered.  For
instance during step (c) we learn that $\varphi_3$ and $\varphi_4$
belong to the same component, but we do not know if this component is
maximal.  We know that a component is maximal when the DFS backtracks
from it.  This happens for the first time at step (e).  At this point
we know that the component containing $\{\ttrue\}$ is maximal, and
because $\lambda(\ttrue)=\top$ we can mark any undetermined state as
winning (green) for the output player.  Marking the
\tikz[baseline=(X.base)]\node[terma,win](X){$\ttrue$}; terminal as
winning will immediately attempt to backpropagate to previous states.
Here, nothing happens, because the predecessors of
\tikz[baseline=(X.base)]\node[terma,win](X){$\ttrue$}; are nodes played
by the input player, who still has other options.  (Our previous
work~\cite{duret.25.ciaa} describes the data-structures we use to
perform backpropagation in an on-the-fly context.)  At step (f), all
successors of $\varphi_3$ and $\varphi_4$ have been explored, so we
know the SCC containing $\{\varphi_3,\varphi_4\}$ is maximal.  Since
$\lambda(\varphi_3)=\lambda(\varphi_4)=\bot$ any undetermined node
in that SCC (in this case, all of them) can be marked (in red) as
losing for the output player.  The DFS continues to explore the successors
of $\varphi_2$ and $\varphi_5$ on steps (g) and (h).  Step (g) has
one interesting bit: when the rightmost node
\tikz[baseline=(X.base)]\node[inode](X){$o$}; (who is played by the
output player) is connected to
\tikz[baseline=(X.base)]\node[terma,win](X){$\ttrue$}; (who is already
known to be winning), then
\tikz[baseline=(X.base)]\node[inode,win](X){$o$}; is automatically
marked as winning too.  Any such determination is automatically
backpropagated, but at step (g) the predecessor of
\tikz[baseline=(X.base)]\node[inode,win](X){$o$}; still has other
choices so nothing happens.  Step (h) computes $\Delta(\varphi_5)$ and
connects $\varphi_5$ to
\tikz[baseline=(X.base)]\node[inode,win](X){$o$};.  Since this is the
only possible choice, $\varphi_5$ can be marked as winning too.  At
step (i), SCC $\{\varphi_5\}$ is backtracked, since
$\lambda(\varphi_5)=\bot$, all its undetermined nodes (there are none
of them) can be marked as losing.  Finally, at step (j) the SCC for
$\{\varphi_1,\varphi_2\}$ is backtracked, and all its undetermined
nodes (all of them) are marked as winning.

In practice the construction can stop as soon as the initial state is
found winning (resp. losing): this shows that its formula is
realizable (resp. unrealizable).  In this example, doing the
construction on-the-fly did not help because the initial state was
only found winning after building the entire automaton.

If the goal is synthesis, not just realizability, the winning nodes
labeled by output variables should also remember the choice that made
them winning.  The green arrows in Figure~\ref{fig:otf} highlight
this.  A winning strategy can be found by ``killing the other choice''
from each winning output node, i.e., redirecting it to a rejecting
$\ffalse$ node, as shown in Figure~\ref{fig:mealy}.  The resulting Mealy
machine is then simplified using techniques described in previous
work~\cite{renkin.22.forte}.

\begin{figure}[tbp]
  \centering
  \begin{tikzpicture}[initial text={},yscale=0.75,baseline=(current bounding box.center)]
    \node[root, accepting] at (3.5,0) (r0) {$\ttrue$};
    \node[root, accepting, initial above] at (2.1,0) (r4) {$\varphi_1$};
    \node[root, accepting] at (2.8,0) (r1) {$\varphi_2$};

    \node[inode] at (2.1,-1) (i13) {$i_1$};
    \node[inode] at (2.8,-1) (i14) {$i_1$};

    \node[inode] at (2.5,-2) (i23) {$i_2$};
    \node[inode] at (3.1,-2) (i24) {$i_2$};

    \node[inode] at (1.8,-3) (o1) {$o$};
    \node[inode] at (2.5,-3) (o2) {$o$};
    \node[inode] at (3.2,-3) (o3) {$o$};

    \node[terma] at (1.4,-4) (t1) {$\varphi_2$};
    \node[terma] at (2.1,-4) (t4) {$\varphi_1$};
    \node[termn,killed] at (2.8,-4) (t5) {$\ffalse$};
    \node[terma] at (3.5,-4) (t0) {$\ttrue$};

    \draw[rlink] (r4) -- (i13);
    \draw[rlink] (r1) -- (i14);
    \draw[rlink] (r0) to[out=-90,in=90] (t0.70);

    \draw[low] (i13) to[bend right=10] (o1);
    \draw[high] (i13) to[bend right=15] (o2);
    \draw[low,killed] (o1) to[out=-45,in=90] (t5.120);
    \draw[high,winedge] (o1) to[out=-135,in=90] (t1);

    \draw[low] (i14) -- (i23);
    \draw[high] (i14) -- (i24);

    \draw[low] (i23) -- (o3);
    \draw[high] (i23) -- (o1);
    \draw[low] (i24) -- (o3);
    \draw[high] (i24) -- (o2);

    \draw[low,killed] (o2) to[out=-45,in=90,looseness=1.6] (t5.90);
    \draw[high,winedge] (o2) to[out=-135,in=90] (t4);

    \draw[low,winedge] (o3) to[out=-45,in=90,looseness=1.1] (t0.100);
    \draw[high,killed] (o3) to[out=-135,in=90] (t5.60);
  \end{tikzpicture}
  \qquad
  \def\conj{}
  \begin{tikzpicture}[mediumautomaton,xscale=0.8,baseline=(current bounding box.center)]
    \node[initial,state] (4) {$\varphi_1$};
    \node[state,right=of 4] (1) {$\varphi_2$};
    \node[state,right=of 1] (0) {$\ttrue$};
    \path[->] (4) edge[loop above] node[auto]{$i_1 / o$} (4)
    (1) edge[loop above] node[auto]{$\bar{i_1}\conj i_2 / o$} (1)
    (0) edge[loop above] node[auto]{$\top/\top$} (0)
    (4) edge[bend right=20] node[above=-1pt]{$\bar{i_1}/ o$} (1)
    (1) edge[bend right=20] node[above=-1pt]{$i_1\conj i_2/ o$} (4)
    (1) edge node[above=-1pt]{$\bar{i_2}/ \bar{o}$} (0);
  \end{tikzpicture}

\caption{Strategy extracted from Figure~\ref{fig:otf}, represented using MTBDDs, or as an incompletely specified Mealy machine.\label{fig:mealy}}
\end{figure}
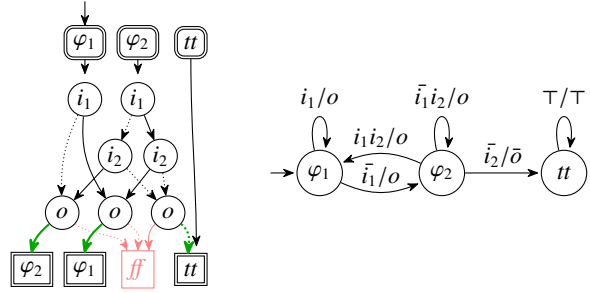

\section{Artifact}\label{app:artifact}

An artifact archived on Zenodo~\cite{duret.26.zenodo} contains:
\begin{itemize}
\item A copy of Spot 2.15.1 that implements the presented techniques.
\item A jupyter notebook showing how to reproduce Figures~\ref{fig:example},
  \ref{fig:otf}, and \ref{fig:mealy}.
\item Source, results, and scripts to reproduce the two benchmarks presented
  in Section~\ref{sec:evaluation}.
\item Scripts to convert the benchmarks results into Figures~\ref{fig:cactus-trans},
  \ref{fig:cactus-time}, \ref{fig:trans-time}, \ref{fig:scatter-trans}, and
  \ref{fig:cactus-mem}.
\item A docker image allowing to execute all of the above, except the
  benchmark of Section~\ref{sec:bench-synthesis} which uses BenchExec
  (cannot run in Docker).
\end{itemize}

\section{Detailed Translation Benchmark}\label{app:trans-benchmark}

This section digs a bit more into the translation benchmark summarized
in Section~\ref{sec:trans-benchmark}.  Recall that the formulas
are of three different kinds.  We use the following names:

\begin{description}
\item[\textsf{scalable}] 310 formulas representing various instances of
  17 scalable patterns~\cite{geldenhuys.06.spin,cichon.09.depcos,geldenhuys.06.spin,kupferman.11.mochart,tabakov.10.rv}.  These patterns are summarized in Table~\ref{tab:scalable-sources}.
\item[\textsf{set}] 76 formulas collected from various works~\cite{dwyer.98.fmsp,somenzi.00.cav,etessami.00.concur,pelanek.07.spin,holevek.04.tr}, these are the last six lines of Table~\ref{tab:obligfreq}.
\item[\textsf{syntcomp}] 108 unique formulas from the synthesis
  competition~\cite{jacobs.22.arxiv}: all instances that are
  syntactic obligations, except for two formulas syntactically equivalent
  to true or false.  This is the first line of
  Table~\ref{tab:obligfreq}.
\end{description}

\begin{table}[tb]
  \caption{Parametric patterns used to evaluate the scalability of the
    translation.  All these patterns can be generated by the
    \texttt{genltl} command of Spot.  The class column shows the
    more precise class the pattern belongs to: (S) syntactic safety,
    (G) syntactic guarantee, or (O)
    syntactic obligation.\label{tab:scalable-sources}}
  \begin{tabular}{lcl}
    \toprule
    name & \llap{cl}ass & formula \\
    \midrule
    \texttt{and-f}~\cite{geldenhuys.06.spin}& G & $\F(p_1)\land\F(p_2)\land\ldots\land\F(p_n)$ \\
    \texttt{ccj-alpha}~\cite{cichon.09.depcos} & G & $\F(p_1\land\F(p_2\land\F(p_3\!\ldots\land\F(p_n))))\land\F(q_1\land\F(q_2\land\F(q_3\!\ldots\land\F(q_n))))$\\
    \texttt{ccj-beta}~\cite{cichon.09.depcos} & G & $\F(p\land\X(p\land\X(p\ldots\land\X(p))))\land\F(q\land\X(q\land\X(q\ldots\land\X(q))))$\\
    \texttt{ccj-beta-prime}~\cite{cichon.09.depcos} & G & $\F(p\land\X(p)\land\X^2(p)\ldots\land\X^n(p))))\land\F(q\land\X(q)\land\X^2(q)\ldots\land\X^n(q))$\\
    \texttt{chomp-mealy$m$} & O & two-player $m\times n$ \href{https://en.wikipedia.org/wiki/Chomp}{chomp game} with Mealy semantics \\
    \texttt{gh-q}~\cite{geldenhuys.06.spin}& O & $(\F(p_1)\lor\G(p_2))\land(\F(p_2)\lor\G(p_3))\land\ldots\land(\F(p_n)\lor\G(p_{n+1}))$ \\
    \texttt{kr-n-delta1} (\ref{app:krndelta1}) & O & formula of size $\mathrm{O}(n)$, with doubly-exponential minimal DBA \\
    \texttt{kr-nlogn-delta1} & O & formula of size $\mathrm{O}(n\log n)$, with doubly-exponential minimal DBA \\
    \texttt{or-g}~\cite{geldenhuys.06.spin}& S & $\G(p_1)\lor\G(p_2)\lor\ldots\lor\G(p_n)$ \\
    \texttt{r-left} & S & $((p_1 \R p_2) \R p_3) \ldots \R p_n$ \\
    \texttt{r-right} & S & $p_1 \R (p_2 \R (\ldots \R p_n))$ \\
    \texttt{tv-f1}~\cite{tabakov.10.rv}& S & $\G(p \limplies (q \lor \X q \lor \X^2q \lor \ldots \lor \X^nq))$ \\
    \texttt{tv-f2}~\cite{tabakov.10.rv} & S & $\G(p \limplies (q \lor \X(q \lor \X(q \ldots \lor \X q))))$ \\
    \texttt{tv-g1}~\cite{tabakov.10.rv}& S & $\G(p \limplies (q \land \X q \land \X^2q \land \ldots \land \X^nq))$ \\
    \texttt{tv-g2}~\cite{tabakov.10.rv} & S & $\G(p \limplies (q \land \X(q \land \X(q \ldots \land \X q))))$ \\
    \texttt{u-left}~\cite{geldenhuys.06.spin}& G & $((p_1 \U p_2) \U p_3) \ldots \U p_n$ \\
    \texttt{u-right}~\cite{geldenhuys.06.spin}& G & $p_1 \U (p_2 \U (\ldots \U p_n))$ \\
    \bottomrule
  \end{tabular}
\end{table}

\begin{figure}[tbp]
  \includegraphics[width=\textwidth]{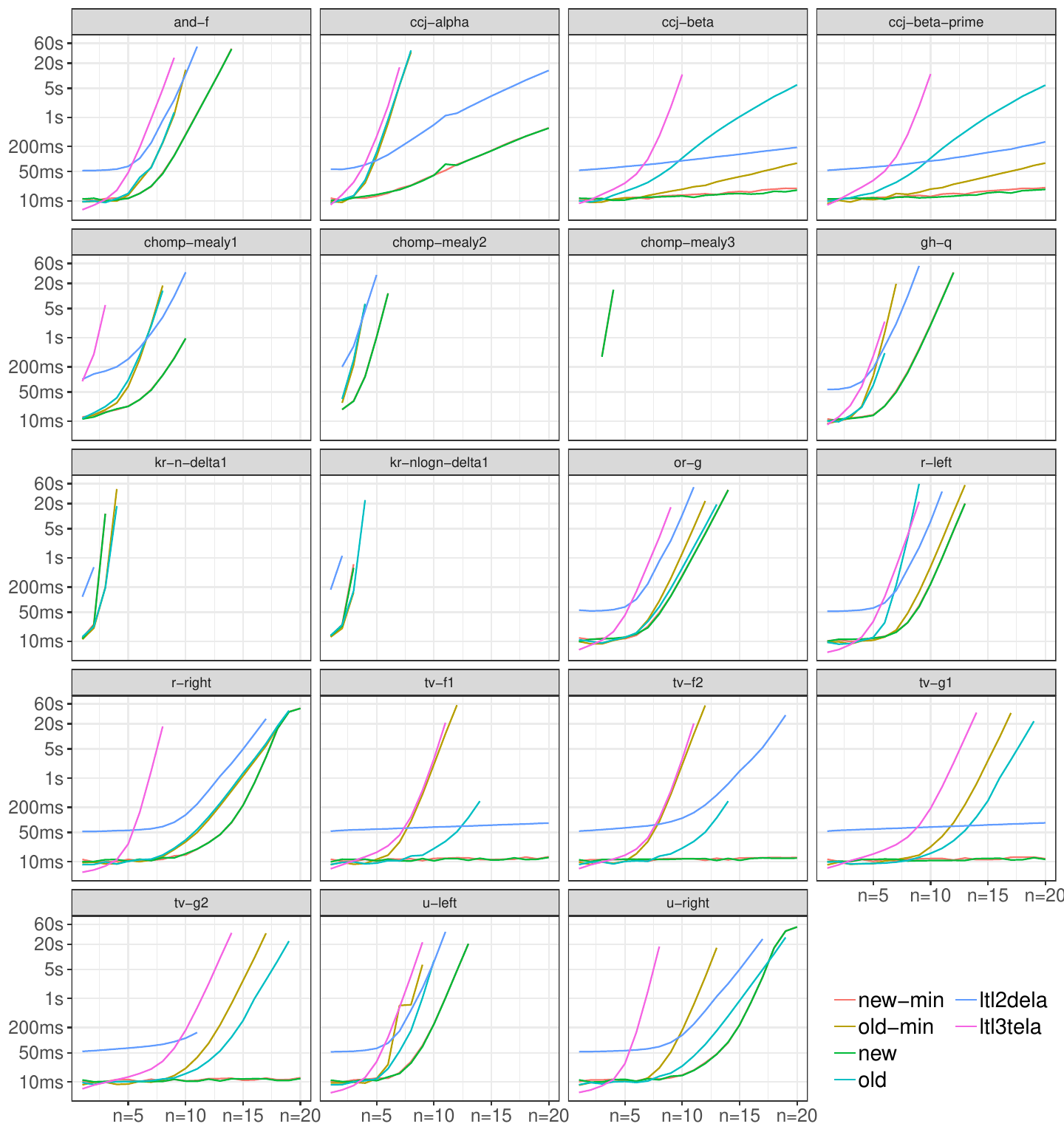}
  \caption{Execution times of the different configurations, on the parametric formulas listed in Table~\ref{tab:scalable-sources}\label{fig:trans-time}}
  \includegraphics[width=\textwidth]{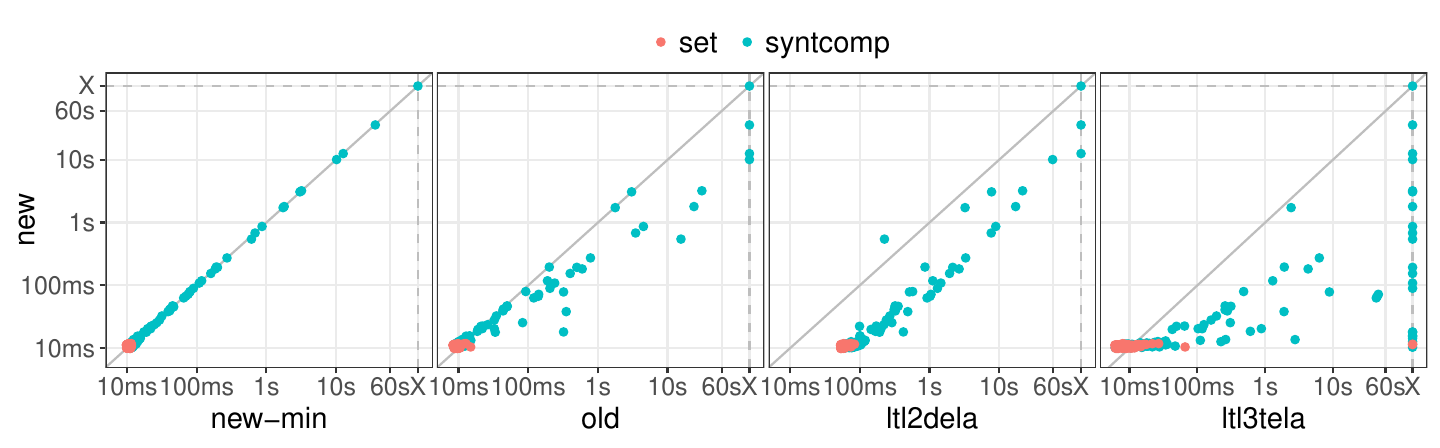}
  \caption{Runtime comparisons over 184 non-scalable formulas from various sources.\label{fig:scatter-trans}}
\end{figure}

Figure~\ref{fig:trans-time} shows the runtime of four different
configurations of \texttt{ltl2tgba} as well as that of
\texttt{ltl2dela} and \texttt{ltl3tela}, over the different patterns
of Table~\ref{tab:scalable-sources}.

In these measurements, we can see that the \pnew{} and
\pnewmin{} pipelines are almost indistinguishable, suggesting
that the minimization is very cheap on these automata.  We looked at
the size of the automata (not shown here), to confirm that on these
scalable patterns the \pnew{} always produces a DWA of optimal
size, making the minimization unnecessary in this case.  The
\pnew{} and \pnewmin{} pipelines also appear to be faster
than everything else, often by some order of magnitude (the time scale
is logarithmic!), with the exception of the \texttt{kr-n-delta1} and
\texttt{kr-nlogn-delta1} formulas.
\marginpar{Cf. App.~\ref{app:krndelta1}.}  On these two families, the
\pold{} pipeline is faster.  Profiling the code for
\texttt{kr-nlogn-delta1} with $n=3$ revealed that on these families
the MTBDD-based translation applying propositional equivalence would build
a 7559-state deterministic automaton that is then minimized to 6207 states.
Couvreur's algorithm also has a form of propositional equivalence, but
it only builds a 1381-state NBA that is then determinized to 6454
states and minimized to 6207 states: a lot of time is saved because
the propositional equivalence does not have to be performed during the
determinization.

In families \textsf{tv-f1}, \textsf{tv-f2}, \textsf{tv-g1}, and
\textsf{tv-g2}, the flat lines for \pnew{} and \pnewmin{} are
mainly due to the fact that propositional equivalence is pushing $\X$ operators
through Boolean formulas, as discussed in Section~\ref{sec:propeq}.

Figure~\ref{fig:scatter-trans} compares \pnew{} against \pnewmin{},
\pold{}, \ltldela{}, and \ltltela{} on 184 syntactic-obligation formulas
from the \textsf{set} and \textsf{syntcomp}.

As can be seen on these plots, the \textsf{set} group of formulas is
not very challenging.

For the \textsf{syntcomp} formulas, we can see again that the new
translation is a clear improvement over \pold{} and the
third-party tools \ltldela{} and \ltltela{}, and that
minimization (\pnewmin) incurs no overhead on these formulas.

It should also be pointed out that \ltltela{} uses a portfolio
approach that involves Spot already: it will first run its own
algorithm, then run Spot's, and return the smallest result.  Its
runtime actually benefits from any improvement to the
translation-pipeline of Spot in general.  In the case of
syntactic obligations, however, since the result of Spot is guaranteed to be minimal, \textsf{ltl3tela}'s first attempt at translating the formula by other means is just unnecessary.

\section{\texttt{kr-n-delta1}}\label{app:krndelta1}

The following formula is a
\href{https://www.lrde.epita.fr/dload/spot/mochart10-fixes.pdf}{fixed version}
of a ``worst-case LTL formula'' for DBA translation, published by
Kupferman \& Rosenberg~\cite{kupferman.11.mochart}.  For a parameter
$n$, the formula has size $\mathrm{O}(n)$ but any equivalent
DBA requires at least $\mathrm{O}(2^{2^n})$ states.

\begin{align}
\# &\land \X\left(a_1 \lor b_1 \lor \$\right)\\
& \land \G\left(\bigwedge_{i=1}^{n-1}\left(\left(a_i \lor b_i\right) \rightarrow \X\left(a_{i+1} \lor b_{i+1}\right)\right)\right) \\
& \land \G\left(\left(a_n \lor b_n\right) \rightarrow \X\left(\# \land \X\left(a_1 \lor b_1 \lor \$ \lor \G \#\right)\right)\right) \\
& \land(\neg \$) \U\left(\$ \land \X\left(\left(a_1 \lor b_1\right) \land \X^k \G \#\right)\right) \label{eq:nonnoblig}\\
& \land \F\left(\# \land \X\left(\neg \# \land\left(\bigvee_{i=1}^n\left(\left(a_i \land \F\left(\$ \land \F a_i\right)\right) \lor\left(b_i \land \F\left(\$ \land \F b_i\right)\right)\right)\right) \U \#\right)\right) \\
& \land \G\left((\# \lor \$) \rightarrow \neg \bigvee_{i=1}^n\left(a_i \lor b_i\right)\right) \land \G(\# \rightarrow \neg \$) \land \G\left(\bigwedge_{i=1}^n\left(a_i \rightarrow \neg b_i\right)\right)
\end{align}

This formula is not a syntactic obligation because equation
\eqref{eq:nonnoblig} contains a $\G$ operator in the right operand of
a $\U$.

However, because proposition $\#$ will eventually occur continuously,
equation \eqref{eq:nonnoblig} can be replaced by the following
equivalent syntactic obligation:
\begin{align*}
  &\land\left((\neg \$) \W\left(\$ \land \X\left(\left(a_1 \lor b_1\right) \land \X^k \G \#\right)\right)\right) \land \F(\#)
\end{align*}

This formula also assumes a letter-alphabet.  To use it in our
context where the alphabet is propositional, we interpret each letter
as a proposition, and add some constraints to ensure that only one
proposition is true at any time, as discussed by Kupferman \&
Rosenberg~\cite{kupferman.11.mochart}.

This modified version is what we call \texttt{kr-n-delta1} in
Table~\ref{tab:scalable-sources}.  \texttt{kr-nlogn-delta1} can be obtained
by a similar modification to the quasilinear formula from the same paper.

\section{\texttt{ltlsynt}'s Architecture}\label{app:ltlsyntarch}

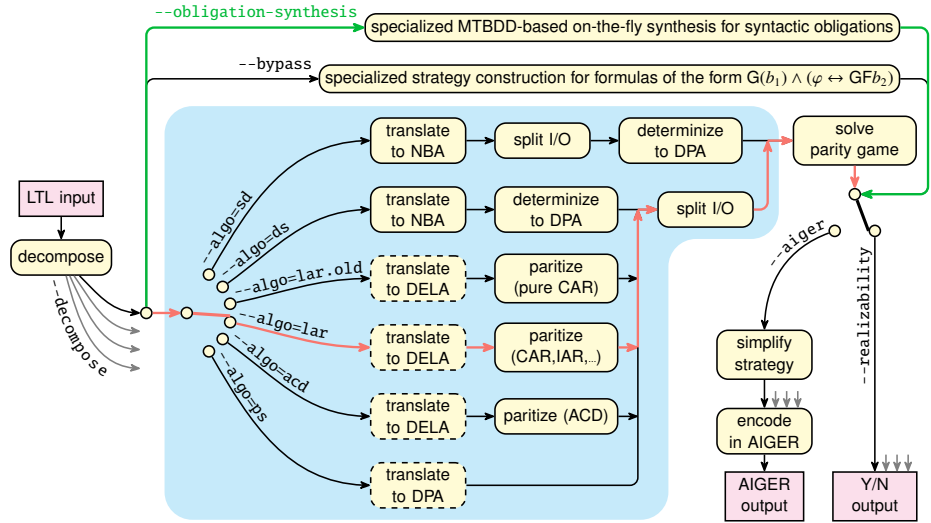
\begin{figure}[tb]
\resizebox{\textwidth}{!}{
  \begin{tikzpicture}[thick,font=\sffamily,>={Stealth[round,bend]},
                      node distance=4mm and 5mm,
                      data/.style={text width=1.3cm,minimum height=2em,align=center,draw,fill=magenta!15},
                      wproc/.style={fill=yellow!20,align=center,draw,rounded corners=2mm},
                      proc/.style={wproc,text width=1.5cm,minimum height=2em},
                      lproc/.style={proc,text width=2cm},
                      choice/.style={circle,fill=yellow!20,inner sep=0,minimum size=2mm,draw},
                      algoopt/.style={postaction={decorate,decoration={text along path, raise=1mm,text={|\ttfamily\small|#1}}}},
                      algooptd/.style={postaction={decorate,decoration={text along path, raise=-1em,text={|\ttfamily\small|#1}}}},
                      outopt/.style={postaction={decorate,decoration={text along path, text align=right,raise=1mm,text={|\ttfamily\small|#1}}}},
                      ]
    \node[proc] (transSD) {translate to NBA};
    \node[proc,right=of transSD] (splitSD) {split I/O};
    \node[lproc,right=of splitSD] (detSD) {determinize to DPA};
    \draw[->] (transSD) edge (splitSD)
              (splitSD) edge (detSD);
    \node[proc,below=of transSD] (transDS) {translate to NBA};
    \node[lproc,right=of transDS] (detDS) {determinize to DPA};
    \node[proc,right=of detDS,xshift=2mm] (splitDS) {split I/O};
    \draw[->] (transDS) edge (detDS)
              (detDS) edge coordinate(middetsplit) (splitDS);
    \node[proc,below=of transDS,dashed] (transLARold) {translate to DELA};
    \node[lproc,right=of transLARold] (paritizeLARold) {paritize (pure CAR)};
    \draw[->] (transLARold) edge (paritizeLARold)
              (paritizeLARold) edge (middetsplit |- paritizeLARold);
    \node[proc,below=of transLARold,dashed] (transLAR) {translate to DELA};
    \node[lproc,right=of transLAR] (paritizeLAR) {paritize (CAR,IAR,{\tiny ...})};
    \draw[->,color1,very thick] (transLAR) edge (paritizeLAR)
    (paritizeLAR) edge (middetsplit |- paritizeLAR);

    \node[proc,below=of transLAR,dashed] (transACD) {translate to DELA};
    \node[lproc,right=of transACD] (paritizeACD) {paritize (ACD)};
    \draw[->] (transACD) edge (paritizeACD)
              (paritizeACD) edge (middetsplit |- paritizeACD);

    \node[proc,below=of transACD,dashed] (transPS) {translate to DPA};
    \draw[rounded corners,->] (transPS) -| (middetsplit);

    \draw[->,color1,very thick] (middetsplit |- paritizeLAR) -- (middetsplit);

    \coordinate (choicecenter) at ($(current bounding box.west) - (3.3cm,0)$);
    \path[draw] (choicecenter) node[choice](algoin){}
          +(60:8mm) node[choice](algosd){}
          +(36:8mm) node[choice](algods){}
          +(12:8mm) node[choice](algolarold){}
          +(-12:8mm) node[choice](algolar){}
          +(-36:8mm) node[choice](algoacd){}
          +(-60:8mm) node[choice](algops){};

    \draw[->,algoopt={-{}-algo=sd}] (algosd) to[out=60,in=180] (transSD);
    \draw[->,algoopt={-{}-algo=ds}] (algods) to[out=36,in=180] (transDS);
    \draw[->,algoopt={-{}-algo=lar.old}] (algolarold) to[out=12,in=180] (transLARold);
    \draw[->,algoopt={-{}-algo=lar},color1,very thick] (algolar) to[out=-12,in=180] (transLAR);
    \draw[->,algoopt={-{}-algo=acd}] (algoacd) to[out=-36,in=180] (transACD);
    \draw[->,algoopt={-{}-algo=ps}] (algops) to[out=-60,in=180] (transPS);
    \draw[ultra thick,color1] (algoin) -- (algolar.north);

    \node[choice,left=of algoin] (bypassin) {};
    \node[proc,left=of bypassin,yshift=1cm,text width=1.6cm] (decomp) {decompose};
    \node[data,above=of decomp] (input) {LTL input};
    \draw[->] (input) -- (decomp);
    \draw[->] (decomp) to[out=-45,in=180] (bypassin);
    \draw[->,gray] (decomp) to[out=-55,in=180] ($(bypassin.west)+(0,-0.33)$);
    \draw[->,gray] (decomp) to[out=-65,in=180] ($(bypassin.west)+(0,-0.66)$);
    \draw[->,gray,algooptd={-{}-decompose}] (decomp) to[out=-75,in=180] ($(bypassin.west)+(0,-1)$);
    \draw[->,color1,very thick] (bypassin) -- (algoin);

    \node[lproc, right=of detSD, xshift=4mm] (solve) {solve\linebreak parity game};

    \node[below=of solve,choice] (reachoice) {};
    \node[below right=of reachoice,choice,xshift=-3mm,yshift=-1mm] (realize) {};
    \node[below left=of reachoice,choice,xshift=3mm,yshift=-1mm] (aiger) {};
    \draw[ultra thick] (reachoice) -- (realize.west);
    \draw[->,color1,very thick] (solve) -- (reachoice);

    \node[data] (yesno) at ($(transPS -| realize)+(0,-2mm)$) {Y/N output};
    \node[data, text width=1.2cm, left=of yesno] (output) {AIGER output};
    \node[proc, above=of output,yshift=-1mm] (encode) {encode in AIGER};
    \node[proc, above=of encode,yshift=+1mm] (minimize) {simplify\linebreak strategy};

    \draw[->] (detSD) -- coordinate(middetsolve) (solve);
    \draw[->,color1,very thick] (middetsplit) -> (splitDS);

    \draw[rounded corners,->,color1,very thick] (splitDS) -| (middetsolve);
    \draw[->,color1,very thick] (middetsolve) -- (solve);

    \draw[<-,outopt={-{}-aiger}] (minimize) to[out=90,in=-160]  (aiger);
    \draw[->] (minimize) -- (encode);
    \draw[gray,->] ($(encode.north)+(2mm,3mm)$) -- ($(encode.north)+(2mm,0)$);
    \draw[gray,->] ($(encode.north)+(4mm,3mm)$) -- ($(encode.north)+(4mm,0)$);
    \draw[gray,->] ($(encode.north)+(6mm,3mm)$) -- ($(encode.north)+(6mm,0)$);
    \draw[<-,outopt={-{}-realizability~}] (yesno) -- (realize);
    \draw[gray,->] ($(yesno.north)+(2mm,3mm)$) -- ($(yesno.north)+(2mm,0)$);
    \draw[gray,->] ($(yesno.north)+(4mm,3mm)$) -- ($(yesno.north)+(4mm,0)$);
    \draw[gray,->] ($(yesno.north)+(6mm,3mm)$) -- ($(yesno.north)+(6mm,0)$);

    \draw[->] (encode) -- (output);

  \begin{scope}[on background layer,overlay]
    \fill[cyan!20,rounded corners=5mm]
    ($(algoin |- transPS.south)+(-.4,-.2)$) |-
    ($(detSD.north -| middetsolve)+(.2,.2)$) |-
    ($(splitDS.south west)+(.2,-.2)$) |- cycle;
  \end{scope}

  \coordinate (solveright) at ($(solve.east)+(2mm,0)$);
  \path (bypassin) -- coordinate(medblue) (solveright);
  \node[wproc,above,xshift=-3mm,yshift=4mm,anchor=south east] (bypass) at (solve.north east) {specialized strategy construction for formulas of the form $\mathsf{G}(b_1) \land (\varphi \leftrightarrow \mathsf{G}\mathsf{F} b_2)$};
  \node[wproc,above,yshift=4mm,anchor=south east] (oblig) at (bypass.north east) {specialized MTBDD-based on-the-fly synthesis for syntactic obligations};
  \draw[rounded corners,->] (bypassin) |- node[at end,anchor=south east]{\texttt{--bypass}} (bypass);
  \draw[rounded corners,->,color3,very thick] (bypassin) |- node[at end,anchor=south east]{\texttt{--obligation-synthesis}} (oblig);
  \draw[rounded corners,->] (bypass) -| (solveright) |- (reachoice);
  \draw[rounded corners,->,color3,very thick] (oblig) -| (solveright) |- (reachoice);
  \end{tikzpicture}}
\caption{Overview of \texttt{ltlsynt}'s architecture, taken from its documentation.\label{fig:ltlsynt-arch}}
\end{figure}

Figure~\ref{fig:ltlsynt-arch} shows the current architecture of
\texttt{ltlsynt}.  The green path at the top of the figure corresponds
to the description of Section~\ref{sec:synthesis}.  It is what is
described as \newsynt{} in Figures~\ref{fig:cactus-time} and
\ref{fig:cactus-mem}.

The large blue area contains multiple ways to convert a specification
into a parity game.  The path highlighted in red is the current
default of \texttt{ltlsynt}.  It will be used if the specialized
techniques do not apply, or if they are disabled (for instance by
setting \texttt{-{}-obligation-synthesis=no}).

The dashed nodes are all using the \texttt{spot::translator} class to
create a deterministic automaton.  They correspond to the dashed block
in Figure~\ref{fig:trans-pipeline}, so they still benefit from the
translation improvement described in Section~\ref{sec:translation}.
In Figures~\ref{fig:cactus-time} and \ref{fig:cactus-mem}, the
\oldtrans{} and \newtrans{} configurations correspond to taking the
red path of Figure~\ref{fig:ltlsynt-arch}, and taking the
\poldmin{} or \pnewmin{} paths in the dashed block of
Figure~\ref{fig:trans-pipeline}.

\section{Memory Usage}\label{app:mem}

\begin{figure}[tb]
  \includegraphics[width=\textwidth]{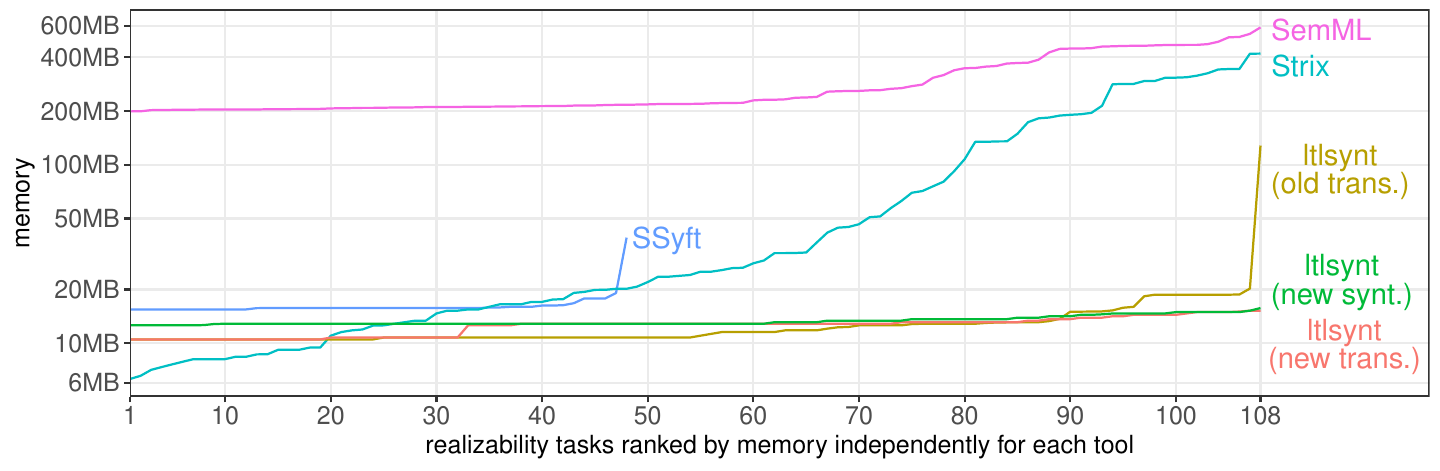}
  \caption{Memory usage for the experiment of
    Figure~\ref{fig:cactus-time}.\label{fig:cactus-mem}}
\end{figure}

Since we used BenchExec for the experiments of
Section~\ref{sec:bench-synthesis}, we can also report about the memory
usage of each tested configuration.  This is plotted in
Figure~\ref{fig:cactus-mem}.

\section{SemML}\label{app:semml}

Our benchmark is not very favorable to \semml, because this tool is
geared toward difficult specifications that are not solved in a few
seconds, using Machine Learning techniques to guide the on-the-fly
construction of the game.

The large runtime observed in Figure~\ref{fig:cactus-time} can be attributed
to the cost of starting a JVM, loading the ML models from files, and
initiating the whole process from a Python script.  For all other
tools, a native binary is called.

\semml{} was the winner of the LTL realizability track of SyntComp'24.
It was outranked by \texttt{ltlsynt} in SyntComp'25 only because of a
small technical issue: its starting script passes the LTL
specification on the command line.  This fails for specifications that
are larger than what a command line can accommodate.  \semml{} failed to
process at least 50 specifications because of this issue.
\texttt{ltlsynt} had no problem reading these large formulas, because
it read them from a pipe.

Those large specifications are not syntactic obligations, so this
was not a problem for our benchmark.

\section{Obligations After Decomposition}\label{app:decomp}

We mention in Section~\ref{sec:decompmention} that some synthesis
specifications can be decomposed into a conjunction of output-disjoint
sub-specifications that can be solved
independently~\cite{finkbeiner.21.nfm,renkin.23.fmsd}.  If the
original specification is not a syntactic obligation, it was not
included in our benchmark.  However, it could be the case that some of the
sub-specifications are syntactic obligations.  Although we did not
evaluate those, they will benefit from the improved game solving.
Table~\ref{tab:decomp} shows the list of such specifications in the
SyntComp benchmark collection.  For instance \texttt{Gamelogic.tlsf} contains 4
sub-specifications, and 3 of those are syntactic obligations.

\begin{table}[tb]
  \centering
  \begin{tabular}{>{\ttfamily}lcc}
    \toprule
    \textrm{file} & obligations/subs \\
    \midrule
    tsl\_paper/SensorInit.tlsf & 1/2 \\
    tsl\_paper/Gamelogic.tlsf & 3/4 \\
    tsl\_paper/KitchenTimerV7.tlsf & 1/2 \\
    tsl\_paper/ModdifiedLedMatrix4X.tlsf & 2/3 \\
    tsl\_paper/KitchenTimerV8.tlsf & 2/3 \\
    tsl\_paper/KitchenTimerV9.tlsf & 1/2 \\
    tsl\_paper/ModdifiedLedMatrix5X.tlsf & 2/3 \\
    tsl\_paper/KitchenTimerV10.tlsf & 1/2 \\
    tsl\_smart\_home\_jarvis/.../Light\_a50cadd7.tlsf & 1/2 \\
    tsl\_smart\_home\_jarvis/.../test4\_4b82b1f4.tlsf & 2/3 \\
    tsl\_smart\_home\_jarvis/.../Room\_a50cadd7.tlsf & 1/3 \\
    tsl\_smart\_home\_jarvis/.../test4\_4b82b1f4\_1.tlsf & 1/2 \\
    tsl\_smart\_home\_jarvis/.../Example1\_d6376bf9.tlsf & 2/3 \\
    tsl\_smart\_home\_jarvis/.../jarvis\_philippe\_484face8\_2.tlsf & 4/5 \\
    tsl\_smart\_home\_jarvis/.../Example.tlsf & 1/2 \\
    \bottomrule
  \end{tabular}
  \caption{List of SyntComp specifications that are not
    syntactic obligations, but that can be decomposed into
    output-disjoint sub-specifications in which there are
    syntactic obligations.\label{tab:decomp}}
\end{table}

\fi
\end{document}